\documentclass[10pt,journal,twoside]{IEEEtran}
\usepackage{graphicx}
\usepackage{epstopdf}
\usepackage{placeins}
\usepackage{array}
\DeclareGraphicsExtensions{.eps,.pdf} 
\graphicspath{ {figures/} }
\usepackage{cite}
\usepackage{balance}
\usepackage{amssymb,amsmath}
\usepackage{subfigure}
\usepackage[ruled,vlined]{algorithm2e}
\usepackage{amssymb,amsmath,mathtools}
\usepackage{amssymb,amsmath}
\usepackage{algorithmic}
\usepackage{color}
\usepackage{multirow}
\usepackage{stackrel,amssymb,amsmath}
\usepackage{empheq}
\usepackage{cases}

\newcolumntype{P}[1]{>{\centering\arraybackslash}p{#1}}

\DeclarePairedDelimiter{\nint}\lfloor\rceil
\DeclarePairedDelimiter{\ceil}{\lceil}{\rceil}
\makeatletter
\newcommand\restartchapters{\par
  \setcounter{chapter}{0}%
  \setcounter{section}{0}%
  \gdef\@chapapp{\chaptername}%
  \gdef\thechapter{\@arabic\c@chapter}}
\makeatother

\def\bq{{\bf q}}

\def\bw{{\bf w}}

\newtheorem{remark}{\it \underline{Remark}}

\newtheorem{lemma}{\it \underline{Lemma}}

\newtheorem{proposition}{\it \underline{Proposition}}

\renewcommand{\algorithmicrequire}{\textbf{Input:}}

\newcommand \ile {\stackrel{\mathclap{\normalfont\mbox{i}}}{\le}}
\newcommand \iieq {\stackrel{\mathclap{\normalfont\mbox{i2}}}{=}}

\newcommand \iiile {\stackrel{\mathclap{\normalfont\mbox{i3}}}{\le}}
\newcommand \iiiile {\stackrel{\mathclap{\normalfont\mbox{i4}}}{\le}}

\newcommand \ifiveeq {\stackrel{\mathclap{\normalfont\mbox{i5}}}{=}}
\newcommand \isixle {\stackrel{\mathclap{\normalfont\mbox{i6}}}{\le}}
\newcommand \isevenle {\stackrel{\mathclap{\normalfont\mbox{i7}}}{\le}}

\makeatletter
\g@addto@macro\normalsize{%
 \setlength\abovedisplayskip{4pt}
 \setlength\belowdisplayskip{4pt}
 \setlength\abovedisplayshortskip{4pt}
 \setlength\belowdisplayshortskip{4pt}
}
\makeatother

\makeatletter
\def\endthebibliography{%
	\def\@noitemerr{\@latex@warning{Empty `thebibliography' environment}}%
	\endlist
}
\makeatother

\begin{document}
\bstctlcite{IEEEexample:BSTcontrol}
\title{\LARGE Satellite- and Cache-assisted UAV: A Joint Cache Placement, Resource Allocation, and Trajectory Optimization for 6G Aerial Networks }
\author{\normalsize
\IEEEauthorblockN{$\text{Dinh-Hieu Tran}, \textit{Graduate Student Member, IEEE}$,  $\text{Symeon Chatzinotas}, \textit{Senior Member, IEEE},$ $\text{and Bj{\"o}rn Ottersten}, \textit{Fellow, IEEE}$ }
\thanks{Dinh-Hieu Tran, Symeon Chatzinotas, and Bj{\"o}rn Ottersten are with the Interdisciplinary Centre for Security, Reliability and Trust (SnT), the University of Luxembourg, Luxembourg. (e-mail: \{hieu.tran-dinh,  symeon.chatzinotas, bjorn.ottersten\} @uni.lu).}
\thanks{Corresponding author: Dinh-Hieu~Tran (e-mail: hieu.tran-dinh@uni.lu).}
\thanks{}
\thanks{}
}
\maketitle
\thispagestyle{empty}
\pagestyle{empty}
\begin{abstract}
This paper considers LEO satellite- and cache-assisted UAV communications for content delivery in terrestrial networks, which shows great potential for next-generation systems to provide ubiquitous connectivity and high capacity. Specifically, caching is provided by the UAV to reduce backhaul congestion, and the LEO satellite supports the UAV's backhaul link. In this context, we aim to maximize the minimum achievable throughput per ground user (GU) by jointly optimizing cache placement, the UAV's resource allocation, and trajectory while cache capacity and flight time are limited. The formulated problem is challenging to solve directly due to its non-convexity and combinatorial nature. To find a solution, the problem is decomposed into three sub-problems: (1) cache placement optimization with fixed UAV resources and trajectory, followed by (2) the UAV resources optimization with fixed cache placement vector and trajectory, and finally, (3) we optimize the UAV trajectory with fixed cache placement and UAV resources. Based on the solutions of sub-problems, an efficient alternating algorithm is proposed utilizing the block coordinate descent (BCD) and successive convex approximation (SCA) methods. Simulation results show that the max-min throughput and total achievable throughput enhancement can be achieved by applying our proposed algorithm instead of other benchmark schemes.
\end{abstract}


\section{Introduction} \label{Introduction}
Although fifth-generation (5G) wireless systems are being deployed around the world \cite{Computerweekly}, the explosive growth of mobile data traffic still poses significant challenges for future networks, i.e., beyond 5G or 6G. It is predicted that individual user data rates will exceed 100 Gbps by 2030, and overall mobile data traffic will reach 5016 exabytes per month \cite{Tariq6G}. To overcome these challenges, the research community is working towards a sixth-generation (6G) system \cite{saad2019vision,Zhen6GLEOSatIoT}. {Notably, the integration of satellite, aerial, and terrestrial networks is promoted as a key factor in providing high-capacity and ubiquitous connectivity  for 6G \cite{saad2019vision,giordani2020non,Mozaffari6G}. }

{Satellite communication (Satcom) has received considerable attention from both industry and academia due to its ability to provide wide-area coverage, e.g., telemedicine, military, satellite-assisted maritime communication, rescue missions, and disaster management system (DMS)} \cite{LiSatUAV2020,ShreeSat2020,LeiSat2020}. Essentially, satellites are installed in geostationary earth orbit (GEO), medium earth orbit (MEO), and/or low earth orbit (LEO), which can complement and support terrestrial communication networks. {Compared to its GEO and MEO counterparts, LEO Satcom operates at much lower altitudes, i.e., from 160 km to 2000 km \cite{ESALEO}, and it provides lower path losses and transmission latency.} Therefore, many projects such as SpaceX, SPUTNIX, OneWeb, and Kepler plan to launch thousands of LEO satellites for providing globally seamless and high throughput communications cooperating with terrestrial communications \cite{DiLEOSatellites}. {Because of these benefits, many works have studied the hybrid LEO satellite-terrestrial communication networks \cite{YouMassiveLEO,ZhangIoTLeO,DengLEO}.} 

Besides many advantages, Satcom is not without limitations. {One satellite can cover a very large area and thus it can improve network performance by acting as a relay or a base station (BS) to provide communication services to isolated areas, e.g., ocean, desert, severe areas.} However, only ground users (GUs) are equipped with expensive high-gain antennas to benefit from satellites. {Other users within satellite coverage cannot take advantage of Satcom's broadband services since they do not have high-gain antennas \cite{LiSatUAV2020}.} In contrast to Satcom, the UAV can fly at much lower altitudes, and thus it can serve as a flying BS to serve GUs due to its flexibility, ease of deployment, and maneuverability. Herein, the satellite can provide a stable backhaul link to the UAV when terrestrial infrastructures are unavailable or have been destroyed after a disaster \cite{saad2019vision,LiSatUAV2020}. Recently, many works investigated the combination of satellite and UAVs in the integrated satellite-UAV-terrestrial network (ISUTN) \cite{LiSatUAV2020,JiaIoTUAVLEO,HuSatUAV2020,HuangSatUAV2020}. {While \cite{LiSatUAV2020,JiaIoTUAVLEO,HuSatUAV2020} only considered dominated LoS channel in UAV communications, \cite{HuangSatUAV2020} proposed a general analytic framework for energy-efficient beamforming scheme of a satellite-aerial-terrestrial network (SATN).} Specifically, satellite-UAV links follow Shadowed-Rician fading while the UAV-user links undergo Rician fading model. Despite noticeable achievements for UAV communications \cite{LiSatUAV2020,JiaIoTUAVLEO,HuSatUAV2020,HuangSatUAV2020}, aforementioned works do not take caching into consideration.

Recent studies have shown that some popular contents are requested repeatedly by users, causing the majority of data traffic \cite{BreslauCaching1999,SajadCache}. {Notably, UAVs only serve GUs by connecting to the content server via a wireless backhaul.} {However, with the explosive growth of the data traffic, the backhaul link can be overloaded due to limited capacity, reducing the user experience. Caching of popular content at the network edge has emerged as a solution to effectively eliminate bottlenecks to backhaul links, especially during peak traffic time. Furthermore, caching techniques can prolong the UAV's lifetime since the UAV can pre-store popular files in its cache to prevent repeated requests from GUs on backhaul links. Therefore, many works have extensively studied the benefits of caching in UAV communications  \cite{JiUAVCache,ChaiUAVcache2020,XuCacheUAV,HieuBack2020, ChengUAVcaching2019,ZhaoUAVcache2018}. Ji et al. \cite{JiUAVCache} maximized minimum throughput among GUs by jointly optimizing cache placement, UAV trajectory, and transmission power. Chai et al. \cite{ChaiUAVcache2020} proposed an online cache-enabled UAV wireless design through jointly optimizing UAV trajectory, transmission power, and caching scheduling. In \cite{XuCacheUAV}, the authors utilized proactive caching at GUs to overcome the UAV endurance issue via jointly optimizing the caching policy, UAV trajectory,  and communication scheduling. Specifically, all the files are cached cooperatively at specified GUs.} If a file is requested, it can be served locally if that GU contains the requested file; otherwise, it can be retrieved from neighboring users through device-to-device (D2D) communications. In \cite{HieuBack2020}, the authors adopted caching and backscatter communication (BackCom) to improve UAV lifetime. {Specifically, they aimed to maximize total throughput via jointly optimizing the backscatter coefficient, dynamic time splitting ratio, and the UAV trajectory with linear and non-linear energy harvesting models. In \cite{ChengUAVcaching2019}, Cheng et al. proposed a new scheme to ensure the security of UAV wireless networks through jointly optimizing the UAV trajectory and time schedule. For GUs with caching capabilities, the UAV could broadcast files to them and prevent eavesdropping. For GUs without caching, the authors maximized the minimum average secrecy rate to guarantee GUs' security. In \cite{ZhaoUAVcache2018}, Zhao et al. investigated caching- and UAV-assisted secure transmission for scalable videos in hyper-dense networks via interference alignment. In this work, both UAVs and small-cell base stations (SBSs) were equipped with caches to store popular videos at off-peak hours. To ensure secure transmission, the idle
SBSs were exploited to generate jamming signal to
disrupt eavesdropping.}

{Unlike previous works that only investigated the integration between satellite and UAV \cite{JiaIoTUAVLEO,HuSatUAV2020,HuangSatUAV2020} or UAV and caching \cite{JiUAVCache,ChaiUAVcache2020,XuCacheUAV,HieuBack2020, ChengUAVcaching2019,ZhaoUAVcache2018}, there are very few works on satellite- and cache-assisted UAV communication networks in the literature \cite{ShushiSatUAVcache,GuSatUAVCaching2020}. Shushi et al. \cite{ShushiSatUAVcache} studied the energy-aware coded caching design cache-enabled satellite-UAV vehicle integrated network (CSUVIN). Notably, the authors in \cite{ShushiSatUAVcache} aimed to minimize the total energy consumption of GEO satellite and UAV but did not solve the problem of maximizing the minimum throughput of GUs.  In \cite{GuSatUAVCaching2020}, the authors investigated a satellite-UAV mobile edge caching system in IoT networks, where the IoT users acted as relays to transfer information from satellite to UAV due to the assuming that there was no direct satellite-to-UAV transmission but stable links existed between the satellite and sensor users.} Further, they considered limited storage capacity for the IoT users but not as a caching model for content-based networks. {In the literature, reference \cite{JiUAVCache} is the most relevant to our work. Nevertheless, \cite{JiUAVCache} considered one-tier data transmission between UAVs to GUs. \textit{In our work, we investigate a two-tier data transmission system model including satellite-to-UAV and UAV-to-user data transmissions.} This is motivated by the fact that the UAV only can pre-store a portion of popular content during off-peak hours and then forward this information to GUs. Due to the limited storage cache size, the UAV cannot store all the files in its cache. In the case that a GU demands content but is not stored in UAV's cache, the UAV will ask the LEO satellite to send the required file on the backhauling link. Moreover, reference \cite{JiUAVCache} assumed that the transmitter only serves up to one requester at a time slot, which is an inefficient method. In this work, we assume that \textit{the UAV can serve multiple GUs simultaneously to improve network performance, i.e., max-min throughput.} Motivated by these observations, our work proposes a novel system model in satellite- and cache-assisted UAV wireless networks that further explores the impact of data transmission latency for backhaul link from satellite-to-UAV due to large distance and limited resource allocation, which makes the problem design even more challenging to solve and has not been investigated before. \textit{To the best of our knowledge, this is the first work that maximizes the GUs' minimum throughput via jointly optimizing UAV's bandwidth, UAV's transmit power, trajectory design, and cache placement in LEO satellite- and cache-assisted UAV wireless networks}.} In summary, our contributions are as follows:
\begin{itemize}
	\item  We propose a novel satellite- and cache-assisted UAV communication system. {Specifically, caching techniques can reduce the burden on the backhaul link during peak times and prolong the UAV's lifetime.} In addition, the LEO satellite helps to deliver requested content that is not cached by the UAV.
	
	\item We formulate an optimization problem to maximize the minimum throughput at GUs, subject to the UAV's limited operation time, the UAV's maximum speed, the UAV's trajectory, limited cache capacity of the UAV, bandwidth allocation for the access link from the UAV to GUs, transmit power allocation at the UAV to transmit data of GU $k$. The formulation belongs to a mixed-integer nonlinear programming (MINLP) problem, which is challenging to solve. Especially, the binary nature of caching decision-making variables makes it more troublesome.	
	
	\item {We solve the problem by decomposing it into three sub-problems:} (1) cache placement optimization with given UAV resources (i.e., bandwidth and transmit power) and trajectory, followed by (2) the UAV bandwidth and transmit power optimization with given cache placement vector and trajectory, and finally, (3) we optimize the UAV trajectory with given cache placement vector and UAV resources. Based on the sub-problems' obtained solutions, we propose an iterative algorithm to solve the non-convex optimization problem by adopting the block coordinate descent (BCD) method and successive convex approximation (SCA) techniques.
	
	\item {The proposed algorithm's effectiveness is shown via simulation results. In particular, our approach yields an improvement up to 26.64$\%$, 79.79$\%$, and 87.96$\%$ in max-min throughput compared to the benchmark scheme one (BS1), benchmark scheme two (BS2), and benchmark scheme three (BS3), respectively.} More specifically, BS1, BS2, and BS3 are designed similarly to the proposed algorithm but with fixed trajectory, no caching capability, and fixed resource allocation, respectively.
\end{itemize}


The rest of the paper is organized as follows. The system model and problem formulation are given in Section~\ref{System_Model}. The proposed iterative algorithm for solving satellite- and cache-aided UAV communications is presented in Section~\ref{sec:3}. Numerical results are depicted in Section~\ref{Sec:Num}, and Section~\ref{Sec:Con} concludes the paper.

\emph{Notation}: Scalars and vectors are denoted by lower-case letters and boldface lower-case letters, respectively. For a set $\mathcal{{K}}$, $|\mathcal{{K}}|$ denotes its cardinality. For a vector $\bf v$, $\left\| \bf v \right\|_1$ and $\left\| \bf v \right\|$  denote its $\ell_1$ and Euclidean ($\ell_2$) norm, respectively. $\mathbb{R}$ represents for the real matrix. $\mathbb{R}^+$ denotes the non-negative real numbers, i.e., $\mathbb{R}^+=\{x \in \mathbb{R}|x \ge 0\}$. $x \sim {\cal{CN}}(0,\sigma^2)$ represents circularly symmetric complex Gaussian random variable with zero mean and variance $\sigma^2$. Finally, $\mathbb{E} [x]$ denotes the expected value of $x$. $\ceil[\big]{x}$ and $\nint{x}$ define a ceiling and nearest integer function of a number $x$, respectively.

\begin{table*}
	\caption{List of Notation}\label{table1}
	\centering
	\begin{tabular}{|p{5cm}||p{10cm}|}
		\hline
		Notations & Descriptions         \\ \hline
		${\cal K}$                & The set of $K$ ground users, $k \in {\cal K}$ \\ \hline
		${\cal N}$                & The set of $N$ time slots, $n \in {\cal N}$  \\ \hline
		${\cal F}$                & The set of $F$ contents in the network, $f \in {\cal F}$  \\ \hline
		$T$                & Total flying time of the UAV  \\ \hline
		$\delta_t$                & Duration of one time slot $n$, with $n \in {\cal N}$ \\ \hline
		$\delta_d$                & Maximum distance that the UAV can travel during time slot $n$, with $n \in {\cal N}$ \\ \hline
		$\bw_k \in {\mathbb{R}^{2 \times 1}}, k \in {\cal K}$ & Cartersian coordinates of the GU $k$ \\ \hline
		$\bq_n \in {\mathbb{R}^{2 \times 1}}, k \in {\cal K}$ & The horizontal location of the UAV during time slot $n$, with $n \in {\cal N}$ \\ \hline
		$H_u$ & The flying altitude of the UAV \\ \hline $d_{1k}$ & Distance between satellite and the UAV
		\\ \hline $d_{2k}^n$ & Distance between the UAV and GU $k$ during time slot $n$
		\\ \hline
		$h_{1k}^n$ & Channel coefficient between satellite and the UAV  to transmit the demanded data for GU $k$ \\ \hline
		$h_{2k}^n$ & Channel coefficient between the UAV and GU $k$ \\ \hline
		$\eta_f$ & Binary variable $\eta_f \in \{0,1\}$ indicates the UAV caches content $f$ or not \\ \hline
		$b_{2k}^n$ & $0 \le b_{2k}^n \le 1$ denotes the allocation bandwidth for the UAV to transmit GU k's requested content during time slot $n$ \\ \hline
		$p_{2k}^n$ & The UAV transmit power to convey GU k's requested content during time slot $n$ \\ \hline
		$\beta_0$ & Channel gain at the reference distance, i.e., $d_0=1$ meter\\ \hline
		$S$ & Maximum number of contents that the UAV can cache \\ \hline
	\end{tabular}
\end{table*}

\begin{figure}[t]
	\centering
	\includegraphics[width=9cm,height=10cm]{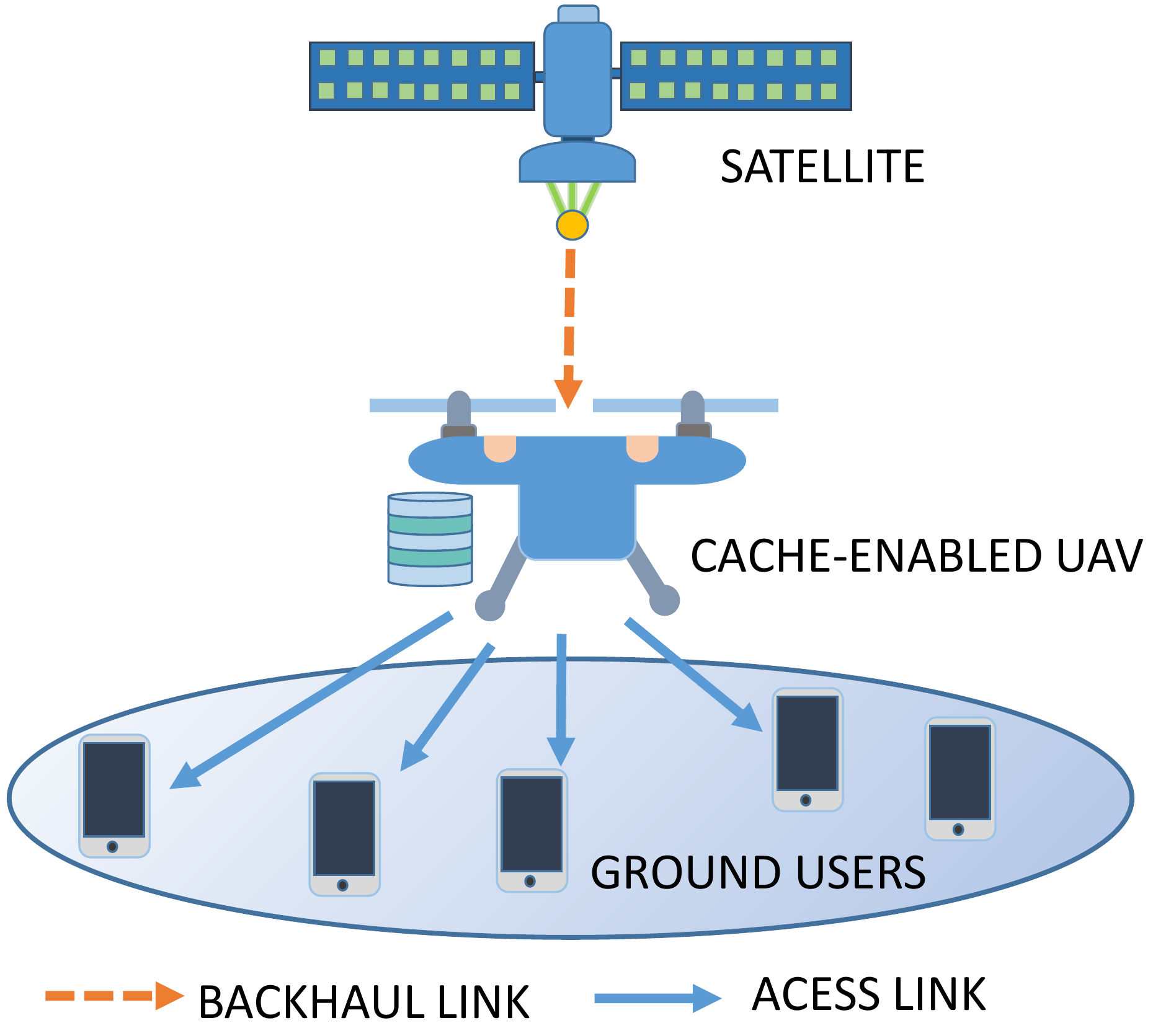}
	\caption { System model: The cache-aided UAV acts as a flying base station to transfer data to ground users, wherein the satellite can provide backhaul link to the UAV. }
	\label{fig:1}   
\end{figure} 
\section{System Model and  Problem Formulation} \label{System_Model}
We consider the downlink of a wireless communication system consisting of a LEO satellite constellation \cite{HuSatUAV2020}, a UAV acts as a flying base station (BS), and a set of $K$ ground users (GUs) denoted by ${\cal K}=\{1,\dots,k,\dots,K\}$. {All GUs are assumed to be within the coverage area of the considered LEO satellite, but it cannot take advantages of the benefits from broadband service since they are not equipped with expensive high-gain antennas \cite{LiSatUAV2020}.} Therefore, the UAV acts as an aerial BS to serve GUs on the access links (ALs), wherein it requests the demanded contents from satellite on the backhaul link (BL). Notably, the satellite backhaul for content delivery is selected because it can broadcast content to multiple UAVs. Further, we assume that the cells of multiple UAVs do not interfere either because they are geographically distant or they use different frequencies; thus, we focus on studying a single UAV. {To reduce the problem of congestion in the backhaul link, the UAV is equipped with a cache \cite{JiUAVcacheD2D,Tran2020FDUAV,JiUAVCache} and it can pre-store a part of popular contents from a ground base station (GBS) during off-peak hours and then forward this information to the GUs. We assume that the flight time of the UAV is $T$ due to the limited battery capacity.} For ease of analysis, $T$ is divided into $N$ equal time slot with duration $\delta_t$, denoted by ${\cal N} \triangleq \{1,\dots,n\dots,N\}$. Since $\delta_t$ is chosen sufficiently small, the distance between the UAV and GU $k$ is assumed to be unchanged during each time slot $n$ \cite{HieuUAVDynamic}. Moreover, we utilize the Cartesian coordinate system, whereas the location of GU $k$ is $\bw_k \in {\mathbb{R}^{2 \times 1}}, k \in {\cal K}$. {For ease of analysis, the UAV's altitude is assumed to be fixed at $H_u$ meters during the flight, e.g., this is imposed by the regulator for safety requirements.} The horizontal location of the UAV during time slot $n$ is denoted by $\bq_n = [x_n,y_n]$. Key notations used in this work are listed in Table \ref{table1}.

\subsection{Caching Model}
The satellite is assumed to access to a set of $F$ contents, denoted by ${\cal F} \triangleq \{1,\dots,f,\dots,F\}$. All contents are assumed to have the same size of $Q$ bits. For different size contents, they can be split into several chunks with equal size \cite{LiSatUAV2020,ThangFDCache}. Therefore, this assumption is applicable in realistic scenarios. In this work, we consider Zipf distribution which is the most popular content popularity model \cite{BreslauCaching1999}. Then, the demanded probability for any content $f$ being requested by GU $k$ is given by
{\begin{align}
	\label{eq:Zipf}
	P_{k,f} = \frac{f^{-\varrho}}{\sum\limits_{i=1}^F i^{-\varrho}},
\end{align}
where $\varrho$ is the Zipf skewness factor with $0 < \varrho < 1$ and $\sum\limits_{f \in {\cal F}}P_{k,f}=1$.}

We assume that the UAV can proactively cache a part of popular contents in its storage during off-peak hours. Specifically, this paper utilizes the placement-then-delivery method. Notably, cache placement is also a part of the optimization problem, which is performed offline before the UAV takes off. Due to the limited storage cache size, the UAV cannot store all files in its cache. {In case a GU requests content but is not cached by the UAV, the UAV will ask the satellite to send the required file on the BL.} Let us define $\boldsymbol{\eta} \triangleq\{\eta_f, \forall f \in {\cal F}\}$ as the cache placement vector of the UAV, where $\eta_f \in \{0,1\}$. Specifically, $\eta_f=1$ means that the UAV caches content $f$, thus the GU $k$ can directly obtain file $f$ from the UAV. Otherwise, $\eta_f^n=0$ indicates that the UAV does not cache file $f$, it thus makes a request from the satellite to serve any GU $k$. {Without loss of generality, once the UAV has received the requested content of the device $k$, then it will begin transmitting the requested content to that user.}\footnote{In this paper, we use a (decode-and-forward) DF relaying technique \cite{hieu2018performance}; therefore, the UAV needs to complete getting the requested data from the satellite before relaying to GU $k$ to ensure the properly encoded data.. } 

{To overcome the problem of wireless BL congestion during peak hours, the UAV is equipped with $SQ$ bits storage capacity, where $S$ is the maximum number of contents that can be stored at the UAV.} Due to the limited cache size, the UAV cannot store all contents, thus we have $S<F$. Moreover, each GU demands content independently. Considering the limited cache storage, the caching constraint can be given by
\vspace{0.2cm}
\begin{align}
	\label{eq:15}
	\sum\limits_{f \in {\cal F}} \eta_f  \le S.
\end{align}


\subsection{Channel Model}
Let $V_{\max}$ be the maximum UAV speed, leading to the following constraints:
\begin{align}
	\label{eq:1}
	\left\| \bq_n - \bq_{n-1} \right\| \le \delta_d = V_{\rm max}{\delta _t},n=1,\dots,N,
\end{align}
\begin{align}
	\label{eq:2}
	\bq_0 = \bq_{\rm I}, \bq_N=\bq_{\rm F},
\end{align}
where $\delta_d$ denotes the maximum distance that UAV can travel during any time slot $n$; $\bq_{\rm I}$ and $\bq_{\rm F}$ are the start point and end point of the UAV.

For analytical convenience, let us denote the satellite, UAV, and GU $k$ by $s$, $u$, and $k$, respectively. {Henceforth, $1k$ and $2k$ are respectively denote the BL (i.e., $s \to u$) and AL (i.e., $u \to k$) to transmit the GU $k$'s demanded data.} Therefore, the distance from $u \to k$ is expressed as

\begin{align}
	\label{eq:3}
	d_{2k}^n = \sqrt{H_u^2+\|\bq_n - \bw_k \|^2}.
\end{align}

We assume that the UAV's altitude is high enough; thus, the channel between satellite and UAV is dominated by LoS link with a limited number of scattered path \cite{HuangSatUAV2020,ZhaoSatUAV2018}. Propagation measurement in \cite{JiangSat2010} mentions that the majority of the total signals transmitted from a satellite to a receiver is predominant by the LoS link. In this regard, the channel coefficient between $s \to u$ at time slot $n$ to transmit GU $k$'s data is expressed as 
$h_{1k}^n \triangleq \beta_0 (d_{1k}^n)^{-\alpha}$, where $\beta_0$ is the channel gain at a reference distance, i.e., $d_0=1$ meter, $\alpha$ is the path loss exponent, $d_{1k}^n$ denote the distance between the UAV and satellite at time slot $n$. Thanks to recent developing techniques, the Doppler shift at the satellite and UAV are assumed to be estimated and compensated accurately \cite{EssaadaliUAVDoppler,NaeemSateDoppler}. {We assume that the satellite starts from an initial point and moves along the x-axis with velocity $v_s$ \cite{HuSatUAV2020}.} Specifically, orbital velocity $v_s$ can be calculated as \cite{Orbital_parameters}
\begin{align}
	\label{eq:sat_vel}
	v_s = \sqrt{\frac{G M}{(R_e + H_s)}},
\end{align}
where $G = $ 6.67259 $10^{-11}$ $\rm m^3kg^{-1}s^{-2}$ is the Universal Constant of Gravitation, $M =$ 5.9736  $10^{24}$ kg is the mass of the Earth, $R_e=$ 6371 km is the radius of the Earth, $H_s$ is altitude of the satellite above the Earth's surface. 


In practice, GUs operate in a variety of environments, e.g., suburban, urban, rural. Therefore, a generalized transmission model comprising of line-of-sight (LOS) and non-line-of-sight (NLOS) components is considered \cite{HieuUAVDynamic,Yaxiong,Tran2020FDUAV}. In this context, a practical channel model is given as follows
\begin{align}
	\label{eq:4}
	h_{2k}^n = \sqrt {\beta_{2k}^n} {\tilde h}_{2k}^n,
\end{align}
where $\beta_{2k}^n$ represents the large-scale fading due to distance or shadowing and ${\tilde h}_{2k}^n$ denotes the small-scale fading channel model. Specifically, $\beta_{2k}^n$ is modeled as
\begin{align}
	\label{eq:5}
	\beta_{2k}^n \triangleq \beta_0 (d_{2k}^n)^{-\alpha} = \frac{\beta_0}{(H_u^2+\|\bq_n - \bw_k \|^2)^{\alpha/2}},
\end{align}
where $d_{2k}^n \triangleq \frac{1}{(H^2+\|\bq_n - \bw_k \|^2)^{1/2} }$ is the distance between the UAV and GU $k$ at time slot $n$.

Refer to \cite{HieuUAVDynamic,Yaxiong,Tran2020FDUAV}, the small-scale fading ${\tilde h}_{uk}^n$ with ${\mathbb{E}}[|{\tilde h}_{2k}^n|^2]=1$, can be expressed as
\begin{align}
	\label{eq:6}
	{\tilde h}_{2k}^n = \sqrt{\frac{K}{1+K}} \bar{h}_{2k}^n + \sqrt{\frac{1}{1+K}} \hat{h}_{2k}^n,
\end{align}
where the Rician factor $K$ is defined as the ratio between the direct path's power and scattered paths' power. Besides, $\bar{h}_{2k}^n$ and $\hat{h}_{2k}^n$ denote the LoS and the NLoS (Rayleigh fading) channel components .

Let $x_{1k}^n$ and $x_{2k}^n$ denote the symbols with unit power  ($\mathbb{E} \left[|x_{1k}^n|^2\right]=1$ and $\mathbb{E} \left[|x_{2k}^n|^2\right]=1$)  sent $s \to u$ and $u \to k$ during time slot $n$, respectively. Consequently, the received signals at the UAV or GU $k$ can be expressed as
\begin{align}
	\label{eq:7}
	y_{ik}^n& = \sqrt{p_{ik}^n} h_{ik}^n x_{ik}^n + n_0, \; \text{with}\; i \in \{1,2\},
\end{align}
{where $p_{1k}^n$ and $p_{2k}^n$ respectively denote the transmit power of LEO satellite and the UAV on the BL link and AL for the data transmission of GU $k$'s at the $n$-th time slot;} $n_0 \sim {\cal{CN}}(0,\sigma^2)$ is the additive white Gaussian noise (AWGN) at the UAV or GU $k$. Notably, the interference caused by the link between $s \to u$ on the link $u \to k$ is neglected due to the limited satellite's transmit power and a very large distance \cite{LiSatUAV2020,HuSatUAV2020}.

Next, the achievable rate (in bps) from $s \to u$ and $u \to k$ to transmit the GU $k$'s data during the $n$-th time slot are respectively calculated as
\begin{align}
	\label{eq:8}
	r_{1k}^n &= B_{1k} \log_2 \left(1+\Gamma_{\rm 1k}^n\right),\\
	\label{eq:9}
	r_{2k}^n &= b_{2k}^n B_{2k} \log_2 \left(1+\Gamma_{\rm 2k}^n \right),
\end{align}
where $\Gamma_{\rm 1k}^n \triangleq \frac{ p_{1k}^n | h_{1k}^n |^2 } { \sigma^2}$, $\Gamma_{\rm 2k}^n \triangleq \frac{ p_{2k}^n {| {\tilde h}_{2k}^n |^2 \beta_0 }}{\left({H_u^2} + {{\left\| {\bq_n - \bw_k} \right\|}^2}\right)^{\alpha/2} \sigma^2}$, $B_{2k}$ denotes the total bandwidth from the UAV to GUs (in Hz); $B_{1k}$ and $b_{2k}^n B_{2k}$ respectively denote the allocated bandwidth for the BL and AL to transmit GU $k$'s data during time slot $n$. Without loss of generality, $b_{2k}^n$ is approximately continuous between 0 and 1. Particularly, since the channel coefficient $h_{2k}^n$ is a random variable, thus the achievable rate $r_{2k}^n$ is also a random variable. Consequently, we pay attention to get the approximated rate, are given as
\begin{align}
	\label{eq:10}
	\mathbb{E}[r_{2k}^n] &= b_{2k}^n B_{2k} \mathbb{E}[ \log_2 \big(1+\Gamma_{2k}^n\big)]).
\end{align}

{Because the closed-form expression of $\mathbb{E}[r_{2k}^n]$ is difficult to obtain, and thus an lower bound for $\mathbb{E}[r_{2k}^n]$ is adopted, which yields the following lemma:}
{\begin{lemma}\label{lemma:1}
	The lower bound of $\mathbb{E}[r_{2k}^n]$ is given as 
	\begin{align}
		\label{eq:Lemma1_2}
		\bar{r}_{2k}^n = b_{2k}^n B_{2k} \log_2 \Bigg(1+  \frac{e^{-E} p_{2k}^n \beta_0} {\left({H_u^2} + {{\left\| {\bq_n - \bw_k} \right\|}^2}\right)^{\alpha/2} \sigma^2} \Bigg).
	\end{align}
\end{lemma}}
\begin{proof}{We define a function $f(z)=\mathbb{E}_{Z}[\log_2(1+e^{\ln z})]$, with $z > 0$, then we have
	\begin{align}
		\label{eq:A5}
		f(z) \ge \log_2\big(1+ e^{\mathbb{E}_Z[\ln z]}\big)\ \triangleq \hat{f}(z).
	\end{align}
	where the inequality \eqref{eq:A5} holds  based on Jensen's inequality for a convex function $\log_2(1+e^{\ln z})$ w.r.t. $z$.}
	
{Let us denote $z \triangleq \frac{ p_{2k}^n {| {\tilde h}_{2k}^n |^2 \beta_0 }}{\left({H^2} + {{\left\| {\bq_n - \bw_n} \right\|}^2}\right)^{\alpha/2} \sigma^2}$. Because $z$ is an exponentially distributed
	random variable, its parameter is $\lambda_Z = (\mathbb{E}[Z])^{-1} = \frac{\zeta_{1k} }{p_{2k}^n \beta_0}$, with $\zeta_{1k} \triangleq \left({H^2} + {{\left\| {\bq_n - \bw_n} \right\|}^2}\right)^{\alpha/2} \sigma^2$. By applying \cite[Eq. 4.331.1]{gradshteyn2014}, $\mathbb{E}_Z[\ln z]$ can be calculated as
	\begin{align}
		\label{eq:A6}
		\mathbb{E}_Z[\ln z] &= \int_{0}^{+\infty} \lambda_Z e^{-z \lambda_Z } \ln{z} dz  = - \big(\ln (\lambda_Z) + E \big), \notag\\
		&= \ln \frac{ p_{2k}^n \beta_0  } {\zeta_{1k}} -E, 
	\end{align}
where $E$ is the Euler–Mascheroni constant, i.e., $E=0.5772156649$ as in \cite[Eq. 8.367.1]{gradshteyn2014}.}
	
By substituting \eqref{eq:A6} into  \eqref{eq:A5} and combining with \eqref{eq:10}, we obtain \eqref{eq:Lemma1_2}. Thus, the Lemma \ref{lemma:1} is proof.
\end{proof}


In this work, the bandwidth allocation on the AL is applied to efficiently utilize the resource. Besides, the spectrum allocation should satisfy: 
\begin{align}
	\label{eq:12}
	&\sum\limits_{k \in {\cal K}} b_{2k}^n \le 1, \forall n,\\
	\label{eq:13}
	&  0 \le b_{2k}^n \le 1,  \forall k,n.
\end{align}

{From \eqref{eq:Lemma1_2}, the achievable throughput of user $k$ (in bits) during flight time $T$ can be given as
\begin{align}
	\label{eq:Throughput}
	\Upsilon_{2k} = \delta_t P_{k,f} \Bigg[n_f \sum\limits_{n_{1k}}^N \bar{r}_{2k}^n + (1-n_f) \sum\limits_{n_{2k}+1}^N \bar{r}_{2k}^n \Bigg],
\end{align}
where the first component represents the case when the UAV pre-stores the file $f$ in its cache, i.e., $n_f=1$, and thus it can transmit directly to GU $k$;  While the second component means that the UAV does not cache the file $f$, i.e., $n_f=0$, it thus starts transmitting data to GU $k$ when it finishes receiving the requested file from the satellite; $n_{1k}$ denotes the first time slot during the flight time $T$; $n_{2k}$ signifies the time slot in which the UAV completely received GU $k$'s requested file from the satellite. }

\subsection{Problem Formulation}
{In this section, we aim to maximize the minimum achievable throughput among GUs through jointly optimizing the cache placement, resources allocation (i.e., bandwidth and transmit power of the UAV to convey the GU $k$'s data at time slot $n$), and the UAV trajectory,} with the assumption that the GUs' locations, the UAV altitude, and the satellite height are known a priori.

Let $\bq \triangleq \{\bq_n, \forall n\}$, ${\bold b} \triangleq \{b_{2k}^n,n \in {\cal N}, k \in {\cal K}\}$, ${\bold p}\triangleq \{p_{2k}^n,n \in {\cal N}, k \in {\cal K}\}$, $\boldsymbol{\eta} \triangleq\{\eta_f, \forall f \in {\cal F}\}$. Based on the above discussions, the optimization problem is formulated as follows
\begin{IEEEeqnarray}{rCl}\label{eq:P1}
	{\cal P}_1: &&{\max_{\bq, {\bf b}, {\bf p}, {\boldsymbol{\eta}} }~~ \min_{\forall k \in {\cal K}}  \Upsilon_{2k}} \IEEEyessubnumber \label{eq:P1:a}\\
	\mathtt{s.t.}~~
	\vspace{0.01cm}
	&& \delta_t \sum\limits_{n = n_{1k}}^{n_{2k}} B_{1k} \log_2 \Bigg(1+  \frac{ p_{1k}^n \beta_0 } { (d_{1k}^n)^\alpha \sigma^2} \Bigg) \notag \ge (1-\eta_f) Q, \notag \\ 
	&&   \forall n, k,
	\IEEEyessubnumber\label{eq:P1:b}\\
	&& \eta_f \in \{0,1\}, \forall f \in {\cal F},
	\IEEEyessubnumber\label{eq:P1:c} \\
	&& \sum\limits_{f \in {\cal F}} \eta_f \le S,  \forall f,
	\IEEEyessubnumber\label{eq:P1:d} \\
	\vspace{0.01cm}
	&& \sum\limits_{k \in {\cal K}} b_{2k}^n \le 1, \forall n, k, \IEEEyessubnumber\label{eq:P1:e} \\
	\vspace{0.03cm}
	&&  0 \le b_{2k}^n \le 1,  \forall n, \IEEEyessubnumber\label{eq:P1:f}\\
	\vspace{0.05cm} 
	&& \left\| {\bq_n-\bq_{n-1}} \right\| \le \delta_d, n=1,\dots,N,  \IEEEyessubnumber\label{eq:P1:g} \\ 
	&& \bq_0= \bq_I, \bq_N= \bq_F, \IEEEyessubnumber\label{eq:P1:h} \\ \vspace{0.05cm}
	&&0 \le \sum\limits_{k \in {\cal K}}  p_{2k}^n \le P_{u}^{\rm max}, \forall n, \IEEEyessubnumber\label{eq:P1:i}
	\vspace{0.2cm}
\end{IEEEeqnarray}

where constraint \eqref{eq:P1:b} guarantees that the non-cached file $f$ is transmitted from satellite to the UAV; \eqref{eq:P1:c} be the binary variable for the UAV to cache file $f$ or not; \eqref{eq:P1:d} means that the total number of contents cached at the UAV should be less than or equal to the storage capacity; \eqref{eq:P1:e} and \eqref{eq:P1:f} signify the bandwidth allocation constraints; \eqref{eq:P1:g} shows the maximum velocity constraint of the UAV; constraint \eqref{eq:P1:h} explains for the start and end points of the UAV; constraint \eqref{eq:P1:i} implies that the transmit power of the UAV is restrained by its maximum power budget. 

It is troublesome to obtain the direct solution of ${\cal P}_1$ since this is a mixed-integer non-linear program (MINLP), which is NP-hard. Specifically, the objective function is a max-min and non-convex function. Besides, the binary nature of constraint \eqref{eq:P1:c} and the non-convexity of \eqref{eq:P1:g}, which introduces intractability. In the next section, an efficient method is proposed to solve it.
\section{Proposed Solution for Solving ${\cal P}_1$}
\label{sec:3}
\subsection{Tractable formulation for ${\cal P}_1$}
It can be seen that the $n_{2k}$ value in problem ${\cal P}_1$ is not determined, thus we cannot solve ${\cal P}_1$ if $n_{2k}$ is a variable. Consequently, we intend to find a specified value for $n_{2k}$. In the case that file $f$ is stored at the UAV, the data transmission time from $s \to u$ equals zero, i.e., $t_{1k}=0$. Therefore, from \eqref{eq:8}, the total transmission time to transmit GU $k$'s data from $s \to u$ can be specified as
{\begin{align}
	\label{eq:16}
	t_{1k} \triangleq  \frac{(1-\eta_f)Q}{\bar{r}_{1k} }.
\end{align}}
where $\bar{r}_{1k} \triangleq \frac{\sum\limits_{\forall n \in {\cal N}} r_{1k}^n}{N}  $ is the average data transmission rate from $s \to u$ during flying time $T$.

Based on \eqref{eq:16}, we have
\begin{align}
	\label{eq:17}
	n_{2k} \triangleq n_{1k} + \ceil[\big]{\frac{t_{1k}}{\delta_t}},
\end{align}

{Without loss of generality, we assume that $n_{2k} \le N$.}

Let us introduce an auxiliary variable $\chi(\bq,{\bold b},{\bold p},\boldsymbol{\eta}) \triangleq \min\limits_{\forall k \in {\cal K}}  \Upsilon_{2k}$ to transform the max-min problem ${\cal P}_1$ into maximization problem, we then have
\begin{IEEEeqnarray}{rCl}\label{eq:P1.1}
	{\cal P}_{1.1}: &&\max_{\bq, {\bf b}, {\bf p}, {\boldsymbol{\eta}}, \chi }~~ \chi(\bq,{\bold b},{\bold p},\boldsymbol{\eta})  \IEEEyessubnumber \label{eq:P1.1:a}\\
	\mathtt{s.t.}~~
	&& {\Upsilon_{2k} \ge \chi, \forall k,}
	\IEEEyessubnumber\label{eq:P1.1:b}\\
	&& \eqref{eq:P1:b}-\eqref{eq:P1:i},
	\IEEEyessubnumber\label{eq:P1.1:c} 
\end{IEEEeqnarray}

To overcome the non-convexity of the problem ${\cal P}_1$, we first decompose ${\cal P}_1$ into three sub-problems, in which we first optimize the cache placement for a given trajectory and resource allocation, followed
by the optimization of resource allocation (i.e., bandwidth and transmit power allocation) for a given trajectory and cache placement, and finally, we perform the UAV trajectory optimization for a given cache placement and resource allocation. Based on the block coordinate descent (BCD) method \cite{HongBCD,PhuBackMEC} and SCA technique \cite{WangSCA}, an efficient iterative algorithm is proposed, while three sub-problems are alternately optimized until convergence.

\subsection{{Subproblem 1: Cache Placement Optimization}} For any given $\bold b^j$, $\bold p^j$, and $\bq^j$, the cache placement $\boldsymbol{\eta}$ can be obtain by solving the following optimization problem:
\begin{IEEEeqnarray}{rCl}\label{eq:P1.2}
	{\cal P}_1^{\boldsymbol{\eta}}: &&\max_{{\boldsymbol{\eta}}, \chi }~~ \chi(\bq^j,\bold b^j,\bold p^j,\boldsymbol{\eta}) \IEEEyessubnumber \label{eq:P1.2:a}\\
	\mathtt{s.t.}~~
	\vspace{0.05cm}
	&& \delta_t \sum\limits_{n = n_{1k}}^{n_{2k}}  \bar{r}_{1k} \ge (1-\eta_f) Q, \forall n, k, 	
	\IEEEyessubnumber\label{eq:P1.2:b0}\\
	&& {\Upsilon_{2k} \ge \chi, \forall k,}
	\IEEEyessubnumber\label{eq:P1.2:b}\\
	\vspace{0.05cm}
	&& \eta_f \in \{0,1\}, \forall f \in {\cal F},
	\IEEEyessubnumber\label{eq:P1.2:c} \\
	&& \sum\limits_{f \in {\cal F}} \eta_f \le S,  \forall f.
	\IEEEyessubnumber\label{eq:P1.2:d} 
\end{IEEEeqnarray}

It can be seen that the problem ${\cal P}_1^{\boldsymbol{\eta}}$ is non-convex due to the binary nature of the \eqref{eq:P1.2:c} constraint. {Thus, we first relax the binary cache placement variable into a continuous one.} As a result, ${\cal P}_1^{\boldsymbol{\eta}}$ can be re-written as
\begin{IEEEeqnarray}{rCl}\label{eq:P1.3}
	{\cal P}_{1.1}^{\boldsymbol{\eta}}: &&\max_{{\boldsymbol{\eta}}, \chi }~~ \chi(\bq^j,\bold b^j,\bold p^j,\boldsymbol{\eta}) \IEEEyessubnumber \label{eq:P1.3:a}\\
	\mathtt{s.t.}~~
	\vspace{0.05cm}
	&& 0 \le \eta_f \le 1, \forall f,
	\IEEEyessubnumber\label{eq:P1.3:b}\\
	\vspace{0.05cm}
	&& \eqref{eq:P1.2:b0},\eqref{eq:P1.2:b}, \eqref{eq:P1.2:d}.
	\IEEEyessubnumber\label{eq:P1.3:c} 
\end{IEEEeqnarray}

It is observed that the problem ${\cal P}_{1.1}^{\boldsymbol{\eta}}$ is convex since the objective function and constraints are convex, i.e., linear. Consequently, the problem ${\cal P}_{1.1}^{\boldsymbol{\eta}}$ can be solved by using standard methods \cite{Boyd}. Nevertheless, since the cache placement variable is relaxed into continuous values between 0 and 1, and thus it does not ensure that the $\eta_f$ can converge to 0 or 1. {Thus, we introduce a penalty function $\mathbb{P}(\eta_f) \triangleq \eta_f (\eta_f -1)$, which is a convex function \cite{Tran2020FDUAV,Boyd}.} 
Hence, the sub-problem $	{\cal P}_{1.1}^{\boldsymbol{\eta}}$ is reformulated as,
\begin{IEEEeqnarray}{rCl}
	\label{eq:P1.4}
	{\cal P}_{1.2}^{\boldsymbol{\eta}}: &&\max_{{\boldsymbol{\eta}}, \chi }~~ \big(\chi(\bq^j,\bold b^j,\bold p^j,\boldsymbol{\eta}) + \kappa \mathbb{P}(\eta_f) \big) \IEEEyessubnumber\label{eq:P14:a}\\
	\mathtt{s.t.}~~
	&&  \eqref{eq:P1.3:b},\eqref{eq:P1.3:c}.
	\IEEEyessubnumber\label{eq:P1.4:b}
\end{IEEEeqnarray}
where $\kappa >0$ is a penalty factor. Especially, the objective in ${\cal P}_{1.2}^{\boldsymbol{\eta}}$ is a difference of concave function, i.e., $f(\eta_f,\chi) \triangleq \big[\chi(\bq^j,\bold b^j,\bold p^j,\boldsymbol{\eta}) - \big(-\kappa \mathbb{P}(\eta_f)\big) \big]$ with convex constraints. As a result, the problem ${\cal P}_{1.2}^{\boldsymbol{\eta}}$ becomes a DC Programming Problem (DCP). To transform ${\cal P}_{1.2}^{\boldsymbol{\eta}}$ become a convex problem, we substitute $\mathbb{P}(\eta_f) $ in the objective by its first-order Taylor expansion at $(j+1)$-th iteration:
{\begin{align}
	\label{eq:23}
	\widehat{\mathbb{P}}(\eta_f) \triangleq \kappa \Big( \mathbb{P}(\eta_f^{(j)}) + \triangledown \mathbb{P}(\eta_f^{(j)}) \big( \eta_f - \eta_f^{(j)} \big)  \Big), 
\end{align}
where
\begin{align}
	\label{eq:24}
	\triangledown \mathbb{P}(\eta_f^{(j)}) = 2\eta_f^{(j)}-1.
\end{align}
Thus, we have
	\begin{align}
		\widehat{\mathbb{P}}(\eta_f) \triangleq \kappa \Big( \eta_f (2\eta_f^{(j)}-1) - (\eta_f^{(j)})^2 \Big), 
\end{align}}

{Consequently, sub-problem ${\cal P}_{1.2}^{\boldsymbol{\eta}}$ is approximately converted to the following linear programming:}
\begin{IEEEeqnarray}{rCl}
	\label{eq:P1.5}
	{\cal P}_{1.3}^{\boldsymbol{\eta}}: &&\max_{{\boldsymbol{\eta}},\chi }~~  \big(\chi(\bq^j,\bold b^j,\bold p^j,\boldsymbol{\eta}) + \widehat{\mathbb{P}}(\eta_f) \big) \IEEEyessubnumber\label{eq:P15:a}\\
	\mathtt{s.t.}~~ 
	&&  \eqref{eq:P1.3:b},\eqref{eq:P1.3:c}. 
	\IEEEyessubnumber\label{eq:P1.5:b}
\end{IEEEeqnarray}

{\begin{remark}
For cache placement optimization in Algorithm \ref{Alg1}, we only solve the relaxed problem \eqref{eq:P1.5} instead of the original problem \eqref{eq:P1.2},
where the binary cache placement variable $\boldsymbol{\eta}$ in the original problem \eqref{eq:P1.2} is relaxed to continuous values between 0 and 1. Notably, if the cache placement variable $\boldsymbol{\eta}$ obtained by Algorithm 1 can converge to binary values, then the relaxation is tight and the solution in \eqref{eq:P1.5} is also a feasible solution of the problem \eqref{eq:P1.2}. Especially, with an appropriate, sufficiently large, and positive value of the penalty parameter $\kappa$, the relaxed problem \eqref{eq:P1.5} is equivalent to the original problem \eqref{eq:P1.2}, as clearly mentioned in \cite[Theorem 2.1]{murray2010algorithm} or in the introduction of \cite{kalantari1982penalty}.
\end{remark}}

\vspace{0.3cm}

\begin{remark}
	The solutions of sub-problem ${\cal P}_{1.3}^{\boldsymbol{\eta}}$ can be obtained by applying standard optimization solver, i.e., CVX \cite{Boyd}. Besides, we observe that constraint \eqref{eq:P1.3:b} is equivalent to \eqref{eq:P1.2:c} when we achieve the optimal solutions of ${\cal P}_{1.3}^{\boldsymbol{\eta}}$. By applying the lower bound result in \eqref{eq:23}, it can be seen that the feasible solutions for ${\cal P}_{1.3}^{\boldsymbol{\eta}}$ are also satisfy ${\cal P}_{1}^{\boldsymbol{\eta}}$, which means that we can obtain at least a locally optimal solution.
\end{remark}

\subsection{{Subproblem 2: UAV Bandwidth and Transmit Power Optimization}}
For any given cache placement $\boldsymbol{\eta}^j$ and UAV trajectory $\bq^j$, $\bold b$ and $\bold p$ can be obtained by solving the following optimization problem:
\begin{IEEEeqnarray}{rCl}\label{eq:P1.6}
	{\cal P}_1^{\bold b,\bold p}: &&\max_{\bf b, \bf p,\chi }~~ \chi(\bq^j,\bold b,\bold p,\boldsymbol{\eta^j}) \IEEEyessubnumber \label{eq:P1.6:a}\\
	\mathtt{s.t.}~~
	&& \zeta_{\bf b, \bf p }^n \ge  \chi, \forall n, k,
	\IEEEyessubnumber\label{eq:P1.6:b} \\
	\vspace{0.05cm}
		&& \delta_t \sum\limits_{n = n_{1k}}^{n_{2k}}  \bar{r}_{1k} \ge (1-\eta_f) Q, \forall n, k, 	
	\IEEEyessubnumber\label{eq:P1.6:b0}\\
	&& \sum\limits_{f \in {\cal F}} \eta_f  \le S, \forall f,
	\IEEEyessubnumber\label{eq:P1.6:c} \\
	\vspace{0.05cm}
	&& \sum\limits_{k \in {\cal K}} b_{2k}^n \le 1, \forall n, k, \IEEEyessubnumber\label{eq:P1.6:d} \\
	\vspace{0.02cm}
	&&  0 \le b_{2k}^n \le 1,  \forall n,k, \IEEEyessubnumber\label{eq:P1.6:e}\\
	\vspace{0.01cm} 
	&&0 \le \sum\limits_{k \in {\cal K}}  p_{2k}^n \le P_{u}^{\rm max}, \forall n, \IEEEyessubnumber\label{eq:P1.6:f}
	\vspace{0.2cm}
\end{IEEEeqnarray}
where 
{\begin{align}
	\zeta_{\bf b, \bf p }^n &\triangleq \delta_t P_{k,f} b_{2k}^n B_{2k} \Bigg[n_f \sum\limits_{n_{1k}}^N \Phi_{\bf p }^n + (1-n_f) \sum\limits_{n_{2k}+1}^N \Phi_{\bf p }^n \Bigg],  \\
	\Phi_{\bf p }^n &\triangleq \log_2 \Bigg(1+  \frac{e^{-E} p_{2k}^n \beta_0} {\left({H_u^2} + {{\left\| {\bq_n - \bw_k} \right\|}^2}\right)^{\alpha/2} \sigma^2} \Bigg).
\end{align}}

Due to the non-convexity of constraints \eqref{eq:P1.6:b}, sub-problem ${\cal P}_1^{\bold b,\bold p}$ is also non-convex. To make ${\cal P}_1^{\bold b,\bold p}$ more tractable, we first introduce slack variable $\bar{\Phi}_{\bf p }^n$ such that
\begin{align}
	\label{eq:27}
	\Phi_{\bf p }^n \triangleq \log_2 \Bigg(1+  \frac{e^{-E} p_{2k}^n \beta_0} {\left({H_u^2} + {{\left\| {\bq_n - \bw_k} \right\|}^2}\right)^{\alpha/2} \sigma^2} \Bigg) \ge \bar{\Phi}_{\bf p }^n.
\end{align}

Consequently, $\zeta_{\bf b, \bf p }^n$ is converted to the following constraint 
{\begin{align}
	\label{eq:28}
	\zeta_{\bf b, \bf p }^n &\ge \delta_t P_{k,f} B_{2k} \notag\\ & \times \Bigg[n_f \sum\limits_{n_{1k}}^N  b_{2k}^n \bar{\Phi}_{\bf p }^n + (1-n_f) \sum\limits_{n_{2k}+1}^N  b_{2k}^n \bar{\Phi}_{\bf p }^n \Bigg] \triangleq \bar{\zeta}_{\bf b, \bf p }^n .
\end{align}}

\begin{lemma}
	\label{lemma:2}
	For any given $b_{2k}^{n,j}$ and $\bar{\Phi}_{\bf p }^{n,j}$ at $(j+1)$-th iteration, the lower bound of $\bar{\zeta}_{\bf b, \bf p }^n$ is expressed as
	{\begin{align}
		\label{eq:29}
		\bar{\zeta}_{\bf b, \bf p}^n &\ge \delta_t P_{k,f} B_{2k} \notag\\ & \times \Bigg[n_f \sum\limits_{n_{1k}}^N  \Theta_{\bf b, \bf p}^n + (1-n_f) \sum\limits_{n_{2k}+1}^N  \Theta_{\bf b, \bf p}^n \Bigg]  \triangleq \widehat{\zeta}_{\bf b, \bf p}^n,
	\end{align}}
	where 
	\begin{align}
		\label{eq:30}
		\Theta_{\bf b, \bf p}^n &\triangleq \frac{\big(b_{2k}^{n,j}+\bar{\Phi}_{\bf p}^{n,j}\big)^2}{4} + \frac{\big(b_{2k}^{n,j}+\bar{\Phi}_{\bf p}^{n,j}\big)}{2} \times \big(b_{2k}^n-b_{2k}^{n,j} \notag\\ &+\bar{\Phi}_{\bf p}^n-\bar{\Phi}_{\bf p}^{n,j}\big)-\frac{\big(b_{2k}^n-\bar{\Phi}_{\bf p}^n\big)^2}{4}.
	\end{align}
	
	\begin{proof}
		It is observed that $b_{2k}^n  \widehat{\Phi}_{2k}^n$ is a non-convex function. To make it tractable, we replace $b_{2k}^n  \bar{\Phi}_{2k}^n$ by an equivalent different of convex function $0.25\Big[\big(b_{2k}^n+\bar{\Phi}_{2k}^n\big)^2-\big(b_{2k}^n-\bar{\Phi}_{2k}^n\big)^2\Big]$. Then, we adopt first-order Taylor expansion to approximate function $0.25\big(b_{2k}^n+\bar{\Phi}_{2k}^n\big)^2$ as follow
		\begin{align}
			\label{eq:31}
			0.25\big(b_{2k}^n+\bar{\Phi}_{2k}^n\big)^2 &\ge \frac{\big(b_{2k}^{n,j}+\bar{\Phi}_{2k}^{n,j}\big)^2}{4} + \frac{\big(b_{2k}^{n,j}+\bar{\Phi}_{2k}^{n,j}\big)}{2} \times \notag \\
			& \big(b_{2k}^n-b_{2k}^{n,j}+\widehat{\Phi}_{2k}^n-\bar{\Phi}_{2k}^{n,j}\big).
		\end{align}
		Consequently, by substituting the lower bound of $0.25\big(b_{2k}^n+\bar{\Phi}_{2k}^n\big)^2$ as in \eqref{eq:31} into $0.25\Big[\big(b_{2k}^n+\bar{\Phi}_{2k}^n\big)^2-\big(b_{2k}^n-\bar{\Phi}_{2k}^n\big)^2\Big]$, the lower bound of $b_{2k}^n \bar{\Phi}_{2k}^n$ can be obtained as \eqref{eq:30}, which finishes the proof.
	\end{proof}
	
\end{lemma}

Bearing all the above developments in mind, the problem ${\cal P}_1^{\bold b,\bold p}$ is re-formulated as
\begin{IEEEeqnarray}{rCl}\label{eq:P1.7}
	{\cal P}_{1.1}^{\bold b,\bold p}: &&\max_{{\bf b}, {\bf p},\chi }~~ \chi(\bq^j,\bold b,\bold p,\boldsymbol{\eta^j}) \IEEEyessubnumber \label{eq:P1.7:a}\\
	\mathtt{s.t.}~~
	&&  \widehat{\zeta}_{\bf b, \bf p}^n  \ge  \chi, \forall n, k,
	\IEEEyessubnumber\label{eq:P1.7:b} \\
	\vspace{0.05cm}
	&&{\log_2 \Bigg(1+  \frac{ e^{-E} p_{2k}^n \beta_0} {\left({H_u^2} + {{\left\| {\bq_n - \bw_k} \right\|}^2}\right)^{\alpha/2} \sigma^2} \Bigg)  \ge \bar{\Phi}_{\bf p}^n,} \notag\\ && {\forall n,k,}
	\IEEEyessubnumber\label{eq:P1.7:f}\\
	\vspace{0.01cm}
	&& \eqref{eq:P1.6:b0}-\eqref{eq:P1.6:f}.
	\IEEEyessubnumber\label{eq:P1.7:e} 
\end{IEEEeqnarray}

\begin{remark}
	It is noted that the problem ${\cal P}_{1.1}^{\bold b,\bold p}$ is a convex problem because the objective function and all constraints are convex. Therefore, it can be numerically solved via standard optimization solvers such as CVX \cite{Boyd}. By applying the lower bound as in Lemma \ref{lemma:2}, the set of feasible solutions for ${\cal P}_{1.1}^{\bold b,\bold p}$ is also a subset feasible solutions for ${\cal P}_1^{\bold b,\bold p}$. Thus, it guarantees that at least a locally optimal solution can be achieved.
\end{remark}

\subsection{{Subproblem 3: UAV Trajectory Optimization}}
For any given bandwidth allocation $\bold{b}^j$, UAV transmit power $\bold{p}^j$, and cache placement $\boldsymbol{\eta}^j$, the UAV trajectory $\bq$ can be obtained by solving following optimization problem
\begin{IEEEeqnarray}{rCl}\label{eq:P1.8}
	{\cal P}_1^{\bq}: &&\max_{\bq,\chi }~~ \chi(\bq,\bold b^j,\bold p^j,\boldsymbol{\eta}^j) \IEEEyessubnumber \label{eq:P1.8:a}\\
	\mathtt{s.t.}~~
	\vspace{0.05cm}
	&& \zeta_{\bq}^n \ge \chi, \forall n,k,
	\IEEEyessubnumber\label{eq:P1.8:b} \\
	&& \delta_t \sum\limits_{n = n_{1k}}^{n_{2k}}  \bar{r}_{1k} \ge (1-\eta_f) Q, \forall n, k, 	
	\IEEEyessubnumber\label{eq:P1.8:b0}\\
	&& \sum\limits_{f \in {\cal F}} \eta_f \le S, \forall f, 
	\IEEEyessubnumber\label{eq:P1.8:c} \\
	&& \left\| {\bq_n-\bq_{n-1}} \right\| \le \delta_d= V_{max}\delta_t, n=1,\dots,N, \notag \\ \IEEEyessubnumber\label{eq:P1.8:d} \\ 
	&& \bq_0= \bq_I, \bq_N= \bq_F, \IEEEyessubnumber\label{eq:P1.8:e} 
\end{IEEEeqnarray}
where
{\begin{align}
	\zeta_{\bq}^n &\triangleq \delta_t P_{k,f} b_{2k}^n B_{2k} \Bigg[n_f \sum\limits_{n_{1k}}^N \Phi_{\bq}^n+ (1-n_f) \sum\limits_{n_{2k}+1}^N \Phi_{\bq}^n \Bigg], \\
	\Phi_{\bq}^n &\triangleq \log_2 \Bigg(1+  \frac{  \psi} {\left({H_u^2} + {{\left\| {\bq_n - \bw_k} \right\|}^2}\right)^{\alpha/2}} \Bigg), \\
	\psi &\triangleq \frac{e^{-E} p_{2k}^n \beta_0}{\sigma^2}.
\end{align} }

The problem ${\cal P}_1^{\bq}$ is a non-convex optimization problem due
to the non-convexity of constraints \eqref{eq:P1.8:b} and the combinatorial of designing UAV trajectory.
To convexify \eqref{eq:P1.8:b}, we introduce slack variable such that $\left({H_u^2} + {{\left\| {\bq_n - \bw_k} \right\|}^2}\right) \le (\nu_{2k}^n)^{2/\alpha}$. Let us denote ${\boldsymbol{\nu}} \triangleq \{\nu_{2k}^n, \forall n,k\}$, then the sub-problem ${\cal P}_1^{\bq}$ can be re-written as 
\begin{IEEEeqnarray}{rCl}\label{eq:P1.9}
	{\cal P}_{1.1}^{\bq}: &&\max_{\bq,{\boldsymbol{\nu},\chi} }~~ \chi(\bq,\bold b^j,\bold p^j,\boldsymbol{\eta}^j)  \IEEEyessubnumber \label{eq:P1.9:a}\\
	\mathtt{s.t.}~~
	&& \left({H_u^2} + {{\left\| {\bq_n - \bw_k} \right\|}^2}\right) \le (\nu_{2k}^n)^{2/\alpha}, \forall n,k,
	\IEEEyessubnumber\label{eq:P1.9:b} \\
	&&  \bar{\zeta}_{\bq}^n  \ge  \chi,  \forall n, k,
	\IEEEyessubnumber\label{eq:P1.7:c} \\
	\vspace{0.05cm}
	&& \eqref{eq:P1.8:b0}-\eqref{eq:P1.8:e},
	\IEEEyessubnumber\label{eq:P1.9:e} 
\end{IEEEeqnarray}
where 
{\begin{align}
	\bar{\zeta}_{\bq}^n &\triangleq \delta_t P_{k,f} b_{2k}^n B_{2k} \Bigg[n_f \sum\limits_{n_{1k}}^N \bar{\Phi}_{\bq}^n + (1-n_f) \sum\limits_{n_{2k}+1}^N \bar{\Phi}_{\bq}^n \Bigg], \\
	\bar{\Phi}_{\bq}^n &\triangleq \log_2 \Bigg(1+  \frac{ \psi} {\nu_{2k}^n } \Bigg).
\end{align}}

It is noted that the problem ${\cal P}_{1.1}^{\bq}$ is still non-convex. To solve this problem, we transform ${\cal P}_{1.1}^{\bq}$ into a convex form by giving following lemmas:
\begin{lemma}
	\label{lemma:3}
	For any given $\nu_{2k}^{n,j} $ at $(j+1)$-th iteration, $\bar{\Phi}_{\bq}^n \triangleq \log_2 \Bigg(1+  \frac{ \psi} {\nu_{2k}^n } \Bigg)$ is lower bounded by
	\begin{align}
		\label{eq:35}
		\bar{\Phi}_{\bq}^n \ge \log_2\big(1+\frac{\psi}{\nu_{2k}^{n,j}}\big) - \frac{\psi(\nu_{2k}^n-\nu_{2k}^{n,j}) } {\nu_{2k}^{n,j} (\nu_{2k}^{n,j} +\psi)\ln 2} \triangleq \hat{\Phi}_{\bq}^n.
	\end{align}
	
	\begin{proof}
		Due to the convexity of the function $f(x)=\log_2(1+1/x)$ with $x>0$. By applying the first-order Taylor expansion to achieve the lower bound of $f(x)$ at given feasible point $x^j$ as
		\begin{align}
			\label{eq:36}
			\log_2 \big(1+\frac{\Omega_1}{x}\big) \ge \log_2\big(1+\frac{\Omega_1}{x^j}\big) - \frac{\Omega_1 (x-x^j) } {x^j (x^j +\Omega_1)\ln 2}.
		\end{align}
		
		By adopting $\Omega_1 = \psi$ and $x=\nu_{2k}^n$, the Lemma \ref{lemma:3} is then proved.
	\end{proof}
	
\end{lemma}


Consequently, we have
{\begin{align}
	\label{eq:38}
	\bar{\zeta}_{\bq}^n \ge \delta_t P_{k,f} B_{2k} \Bigg[n_f \sum\limits_{n_{1k}}^N \hat{\Phi}_{\bq}^n + (1-n_f) \sum\limits_{n_{2k}+1}^N \hat{\Phi}_{\bq}^n \Bigg]  \triangleq \hat{\zeta}_{\bq}^n.
\end{align}}


Bearing all the above discussions in mind, the problem ${\cal P}_{1.1}^{\bq}$ is re-formulated as
\begin{IEEEeqnarray}{rCl}\label{eq:P1.10}
	{\cal P}_{1.2}^{\bq}: &&\max_{\bq,{\boldsymbol{\nu},\chi} }~~ \chi(\bq,\bold b^j,\bold p^j,\boldsymbol{\eta}^j) \IEEEyessubnumber \label{eq:P1.10:a}\\
	\mathtt{s.t.}~~
	&& \eqref{eq:P1.8:b0}-\eqref{eq:P1.8:e}, \eqref{eq:P1.9:b}, 
	\IEEEyessubnumber\label{eq:P1.10:b} \\
	&&  \hat{\zeta}_{\bq}^n \ge  \chi,  \forall n, k,
	\IEEEyessubnumber\label{eq:P1.10:c}  
\end{IEEEeqnarray}

Because the objective function and all constraints of ${\cal P}_{1.2}^{\bq}$ are convex, thus problem ${\cal P}_{1.2}^{\bq}$ can be directly solved by using standard methods \cite{Boyd}. Consequently, we propose an alternating algorithm based on three sub-problem solutions described in Algorithm \ref{Alg1}.
\begin{algorithm}[h]
	\begin{algorithmic}[1]
		\label{Alg1}
		\protect\caption{Proposed Iterative Algorithm to Solve ${\cal P}_1$}
		\label{alg_1}
		\global\long\def\algorithmicrequire{\textbf{Initialization:}}
		\REQUIRE  Set $j:=0$ and initialize ${\bold{b}}^j$, ${\bold{p}}^j$, and ${\bq}^j$.
		\vspace{0.05cm}
		\REPEAT
		\vspace{0.05cm}
		\STATE Solve ${\cal P}_1^{\boldsymbol{\eta}}$ for given $\{ {\bold{b}}^j,{\bold{p}}^j, {\bq}^j\}$ and denote the optimal solution as $\boldsymbol{\eta}^{j+1}$.
		\vspace{0.05cm}
		\STATE Solve ${\cal P}_1^{\bold b, \bold p}$ for given $\{{\boldsymbol{\eta}}^{j+1}, {\bq}^j\}$ and denote the optimal solution as $\bold b^{j+1}$ and $\bold p^{j+1}$.
		\vspace{0.05cm}
		\STATE Solve ${\cal P}_1^{\bq}$ for given $ \{{\boldsymbol{\eta}}^{j+1}, \bold b^{j+1}, \bold p^{j+1}\}  $ and denote the optimal solution as $\boldsymbol{\bq}^{j+1}$.
		\vspace{0.05cm}
		\STATE Set $j:=j+1.$
		\vspace{0.05cm}
		\UNTIL Convergence \\
\end{algorithmic} \end{algorithm}

\subsection{Convergence and Complexity Analysis}
\subsubsection{Convergence Analysis}
\begin{proposition}
	\label{proposition_1}
	From the proposed iterative Algorithm \ref{Alg1}, we can obtain at least a locally optimal solution.
\end{proposition}

\begin{proof}
	Let us define $\Xi({\boldsymbol{\eta}}^j, \bold b^j, \bold p^j, {\bq}^j)$, $\Xi^{\bold{b,p}}_{\rm lb}({\boldsymbol{\eta}}^j, \bold b^j, \bold p^j, {\bq}^j)$, and $\Xi^{\bq}_{\rm lb}({\boldsymbol{\eta}}^j, \bold b^j, \bold p^j, {\bq}^j)$ as the objective values of ${\cal P}_1$, ${\cal P}_{1.1}^{\bold b, \bold p}$, and ${\cal P}_{1.2}^{\bq}$ at the $j$-th iteration. In the $(j+1)$-th iteration, in step 2 of Algorithm \ref{Alg1}, we have
	\begin{align}
		\label{eq:54}
		\Xi({\boldsymbol{\eta}}^j, \bold b^j, \bold p^j, {\bq}^j) \ile \Xi({\boldsymbol{\eta}}^{j+1}, \bold b^j, \bold p^j, {\bq}^j).
	\end{align}
	The inequality $(i)$ holds because $\boldsymbol{\eta}^{j+1}$ is a optimal solution of ${\cal P}_1^{\boldsymbol{\eta}}$. After that, in step 3 of Algorithm \ref{Alg1}, we have
	\begin{align}
		\label{eq:55}
		\Xi({\boldsymbol{\eta}}^{j+1}, \bold b^j, \bold p^j, {\bq}^j) &\iieq \Xi^{\bold b, \bold p}({\boldsymbol{\eta}}^{j+1}, \bold b^j, \bold p^j, {\bq}^j) \notag\\ &\iiile  \Xi^{\bold b, \bold p}({\boldsymbol{\eta}}^{j+1}, \bold b^{j+1}, \bold p^{j+1}, {\bq}^j) \notag\\ &\iiiile  \Xi({\boldsymbol{\eta}}^{j+1}, \bold b^{j+1}, \bold p^{j+1}, {\bq}^j).	
	\end{align}
	The inequality $(i2)$ holds because first-order Taylor approximation at points $\bold b^j$ and $\bold p^j$ is tight as in constraint \eqref{eq:27}, \eqref{eq:28}, and \eqref{eq:29} \cite{WangSCA}. Moreover, the inequality $(i3)$ holds since $\bold b^{j+1}$ and $\bold p^{j+1}$ are optimal solutions of ${\cal P}_{1.1}^{\bold b, \bold p}$. Then, the inequality $(i4)$ holds because the objective value of ${\cal P}_{1.1}^{\bold b, \bold p}$ is a lower bound to that of ${\cal P}_{1}^{\bold b, \bold p}$ at given points $\bold b^{j+1}$ and $\bold p^{j+1}$. Further, in step 4, we have
	\begin{align}
		\label{eq:57}
		\Xi({\boldsymbol{\eta}}^{j+1}, \bold b^{j+1}, \bold p^{j+1}, {\bq}^j) &\ifiveeq \Xi^{\bq}_{\rm lb}({\boldsymbol{\eta}}^{j+1}, \bold b^{j+1}, \bold p^{j+1}, {\bq}^j) \notag\\ &\isixle \Xi^{\bq}_{\rm lb}({\boldsymbol{\eta}}^{j+1}, \bold b^{j+1}, \bold p^{j+1}, {\bq}^{j+1}) \notag\\&\isevenle \Xi({\boldsymbol{\eta}}^{j+1}, \bold b^{j+1}, \bold p^{j+1}, {\bq}^{j+1}).	
	\end{align}
	The equality $(i5)$ holds because the first-order Taylor expansion at given point $\bq^j$ as in \eqref{eq:35} and \eqref{eq:38} are tight, and the inequality $(i6)$ holds since ${\bq}^{j+1}$ is a optimal solution of ${\cal P}_{1.2}^{\bq}$. Furthermore, the inequality $(i7)$ holds since the optimal value of ${\cal P}_{1.2}^{\bq}$ is a lower bound of ${\cal P}_{1}^{\bq}$ at given ${\bq}^{j+1}$. 
	
	From \eqref{eq:54} to \eqref{eq:57}, we conclude that $\Xi({\boldsymbol{\eta}}^j, \bold b^j, \bold p^j, {\bq}^j) \le \Xi({\boldsymbol{\eta}}^{j+1}, \bold b^{j+1}, \bold p^{j+1}, {\bq}^{j+1})$, which shows that the objective value of ${\cal P}_{1}$ is non-decreasing as the number of iteration increases. Further, the objective value of ${\cal P}_{1}$ is limited by an upper bound value due to the maximum transmit power $P_u^{\rm max}$, restricted traveling time $T$, and maximum bandwidth allocation for each GU. Therefore, this guarantees for the convergence of Algorithm \ref{Alg1}.
\end{proof}



\subsubsection{Complexity Analysis} We analyze the worst-case complexity of Algorithm \ref{Alg1}. Because ${\cal P}_{1.3}^{\boldsymbol{\eta}}$ is a linear programming and it can be solved by using interior method with computational complexity is ${\cal O}\big[L_1\big( (F+1+NK)^{0.5}(N+3)^3\big)\big]$, where $L_1$, $(N+3)$, and $(F+1+NK)$ denote the number of iterations required to update the cache placement, the number of variables, the number of constraints, respectively \cite{ben2001lectures}. Further, the problem ${\cal P}_{1}^{\bold b, \bold p}$ includes $(7NK+2K+3N+1)$ linear or quadratic constraints and $(5NK+1)$ variables, thus its complexity is ${\cal O}\big[L_2\big( (7NK+2K+3N+1)^{0.5}(5NK+1)^3\big)\big]$, where $L_2$ is the number of iterations to update bandwidth and power allocation. Next, the problem ${\cal P}_{1}^{\bq}$ includes $(3NK+2K+3N+2)$ linear or quadratic constraints and $(3NK+2N+1)$ variables, thus its complexity is ${\cal O}\big[L_3\big( (3NK+2K+3N+2)^{0.5}(3NK+2N+1)^3\big)\big]$, where $L_3$ is the number of iterations to update the UAV trajectory. Then, the overall complexity of Algorithm \ref{Alg1} is ${\cal O} \big[L_4\big(L_1\big( (F+1+NK)^{0.5}(N+3)^3\big) + L_2\big( (7NK+2K+3N+1)^{0.5}(5NK+1)^3\big) +  L_3\big( (3NK+2K+3N+2)^{0.5}(3NK+2N+1)^3\big)\big)\big]$ where $L_4$
is the number of iterations until convergence.

\begin{figure}[t]
	\centering
	\includegraphics[width=9.5cm,height=10cm]{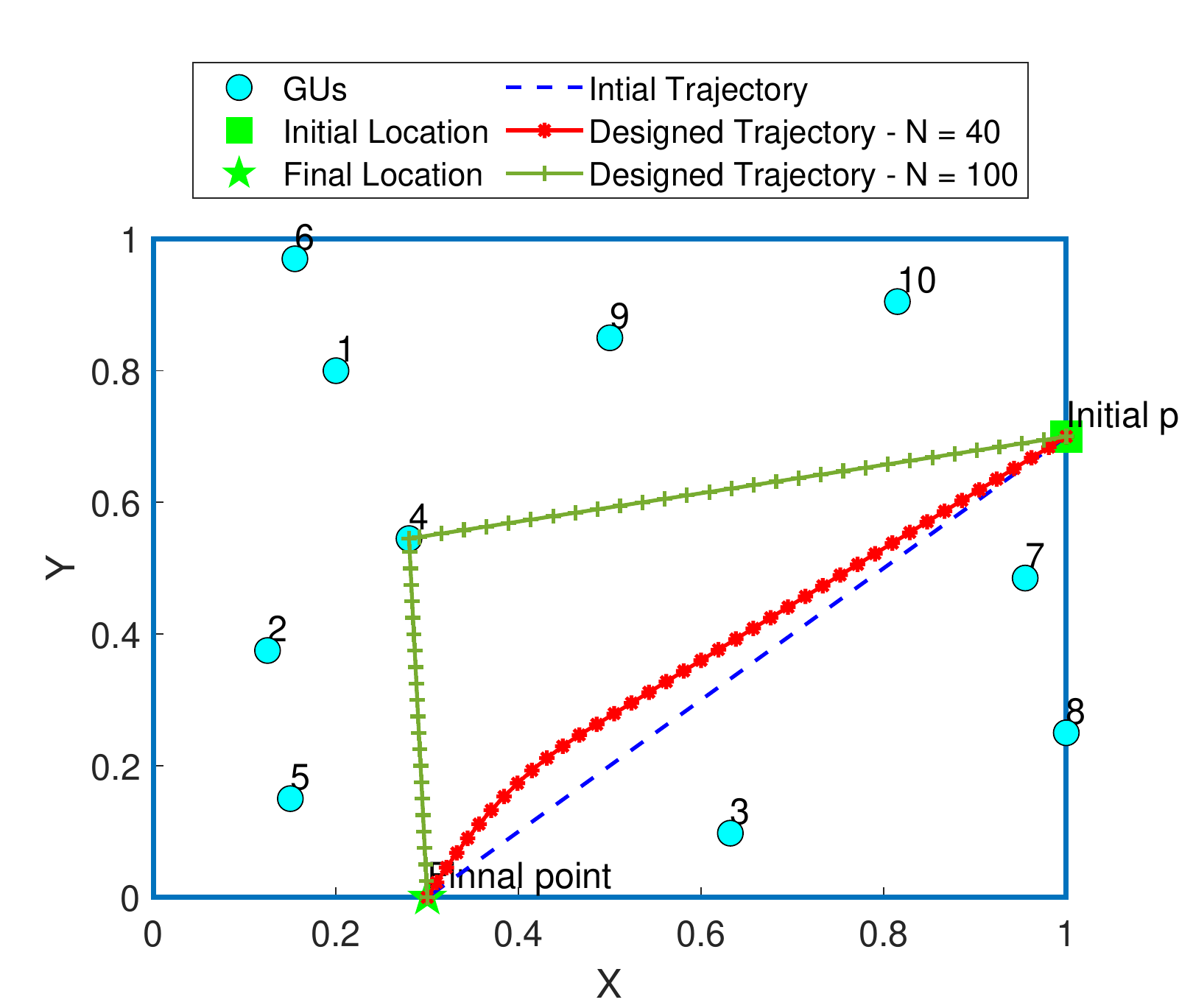}
	\caption{Geometry distribution of GUs and the UAV trajectories. }
	\label{fig:2}   
\end{figure}

\begin{figure*}[t]
	\centering   
	\subfigure[{Max-min throughput}] {\label{fig:3a}\includegraphics[width=9cm,height=7.5cm]{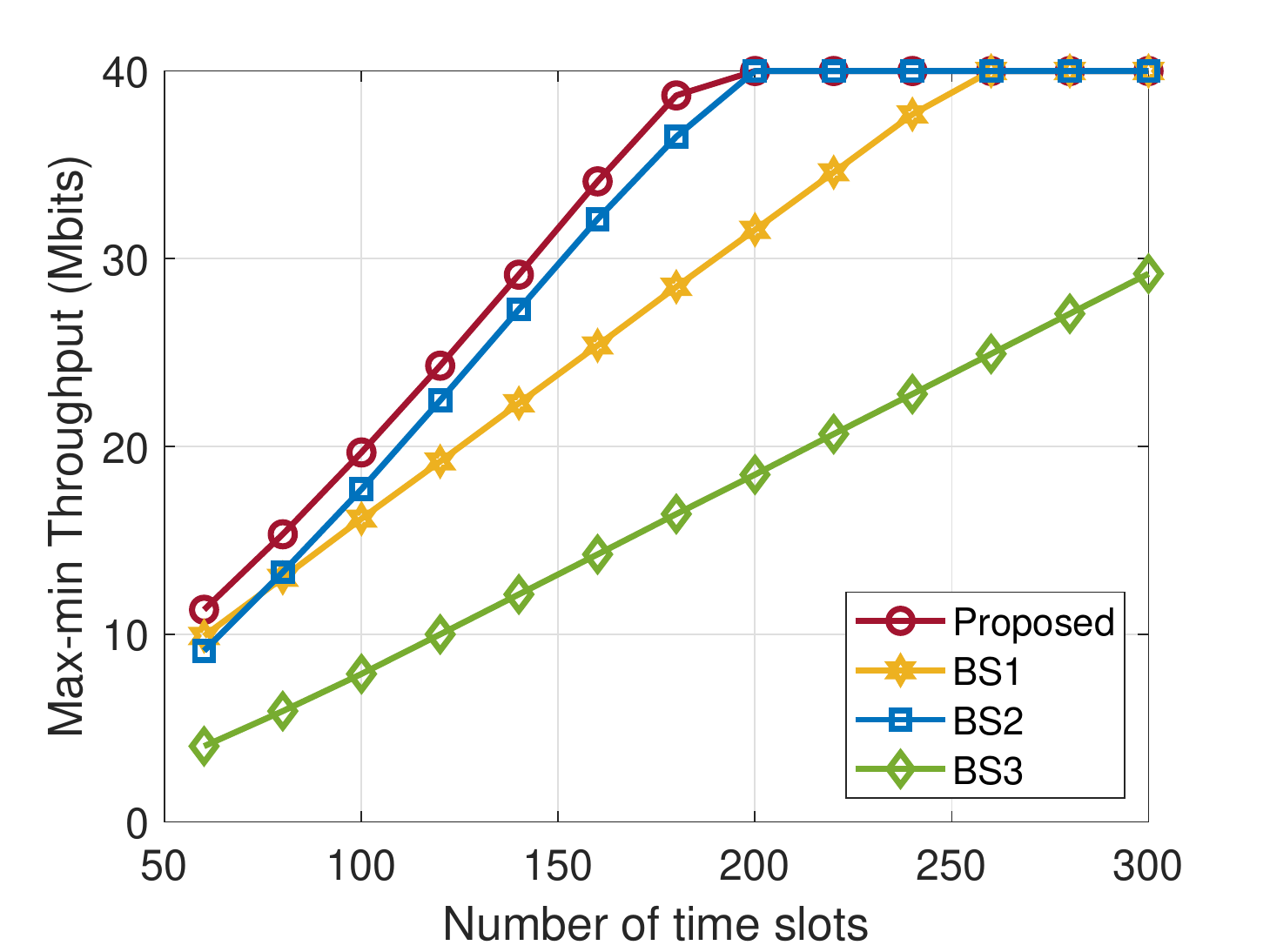}}
	\subfigure[{Total throughput}] {\label{fig:3b}\includegraphics[width=9cm,height=7.5cm]{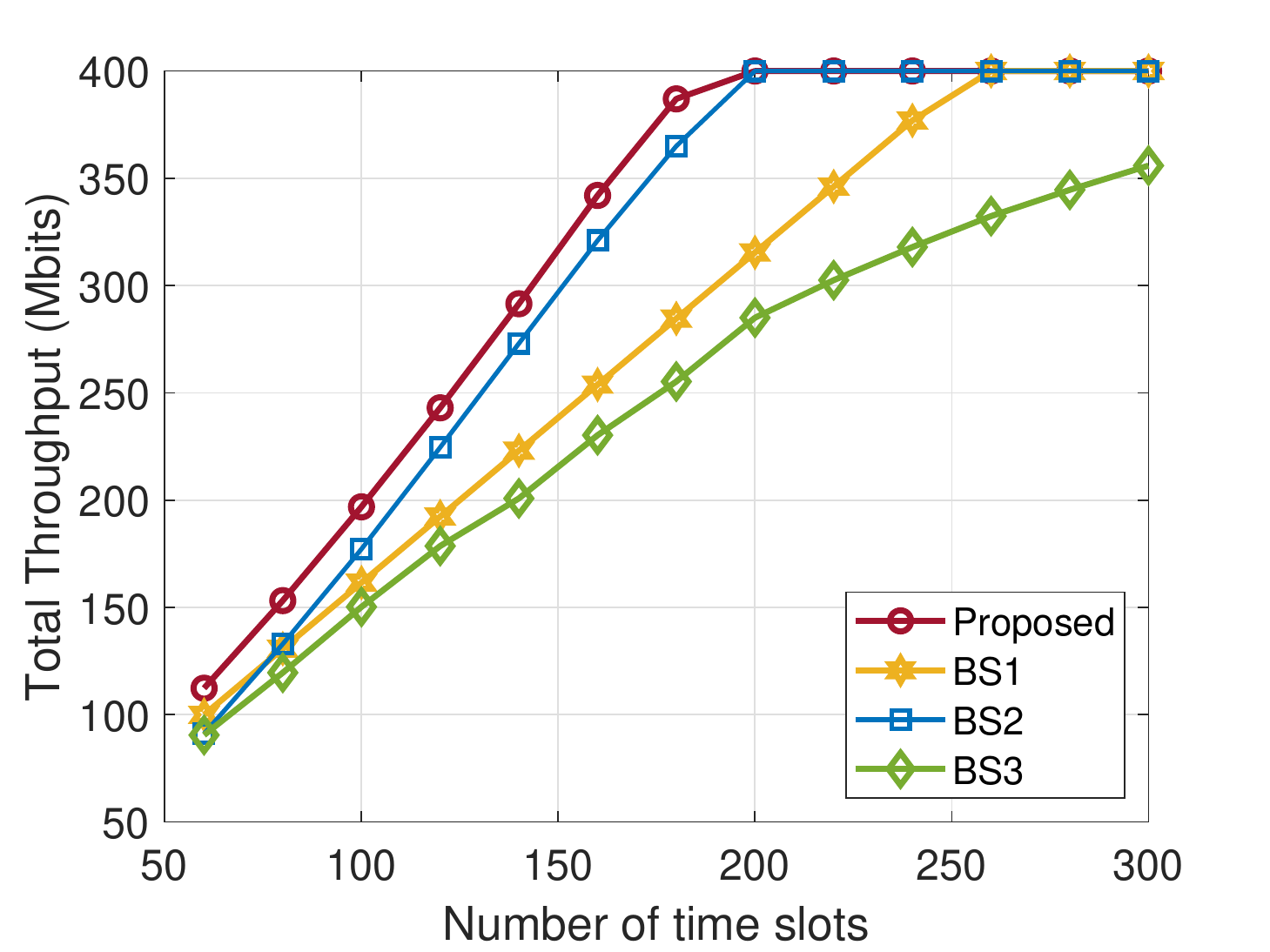}}
	\caption{{Max-min throughput and total throughput vs. number of time slots.}}
	\label{fig:3}
\end{figure*}

\begin{figure*}[t]
	\centering   
	\subfigure[{Max-min throughput}] {\label{fig:4a}\includegraphics[width=9cm,height=7.5cm]{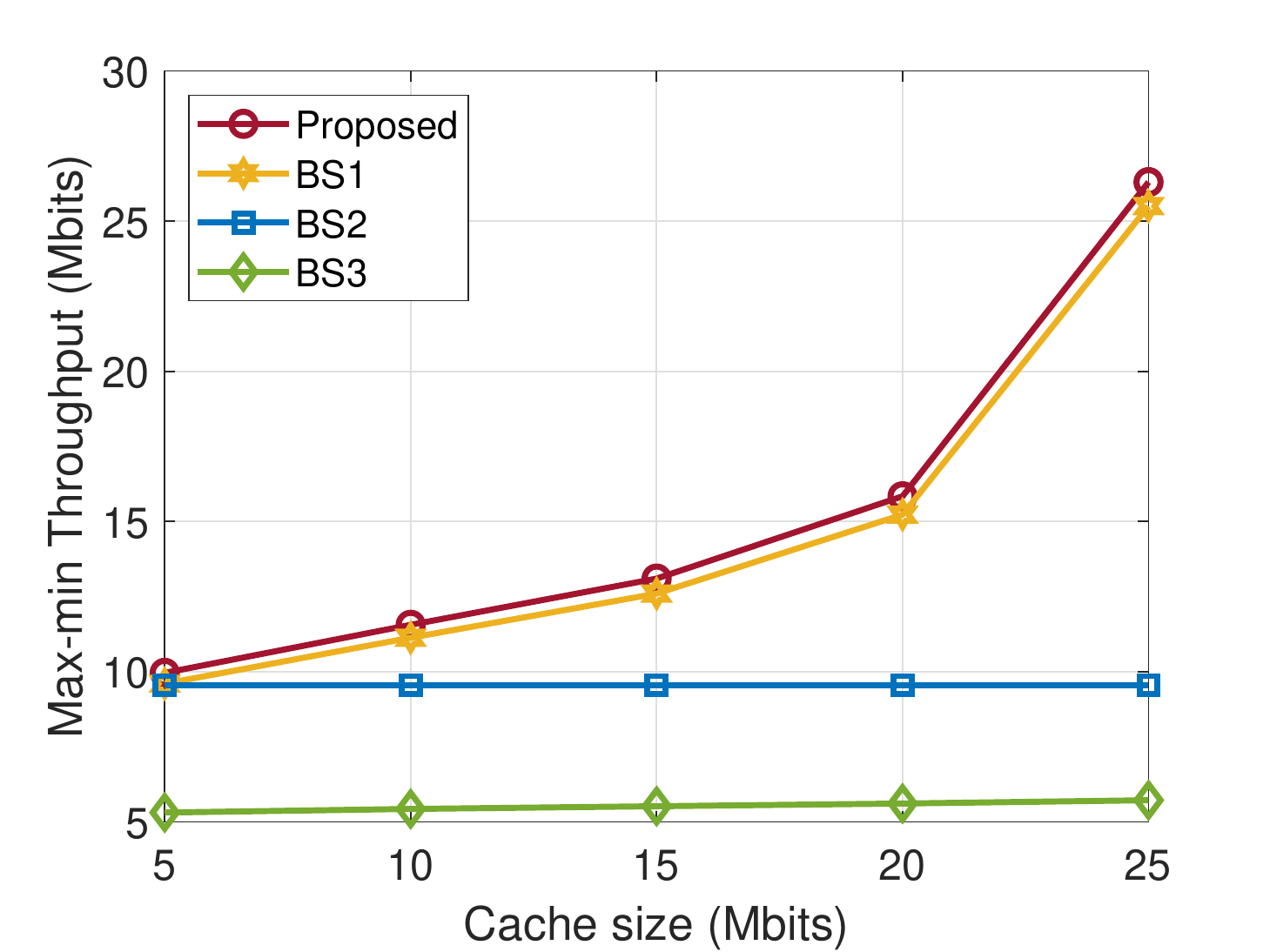}}
	\subfigure[{Total throughput}] {\label{fig:4b}\includegraphics[width=9cm,height=7.5cm]{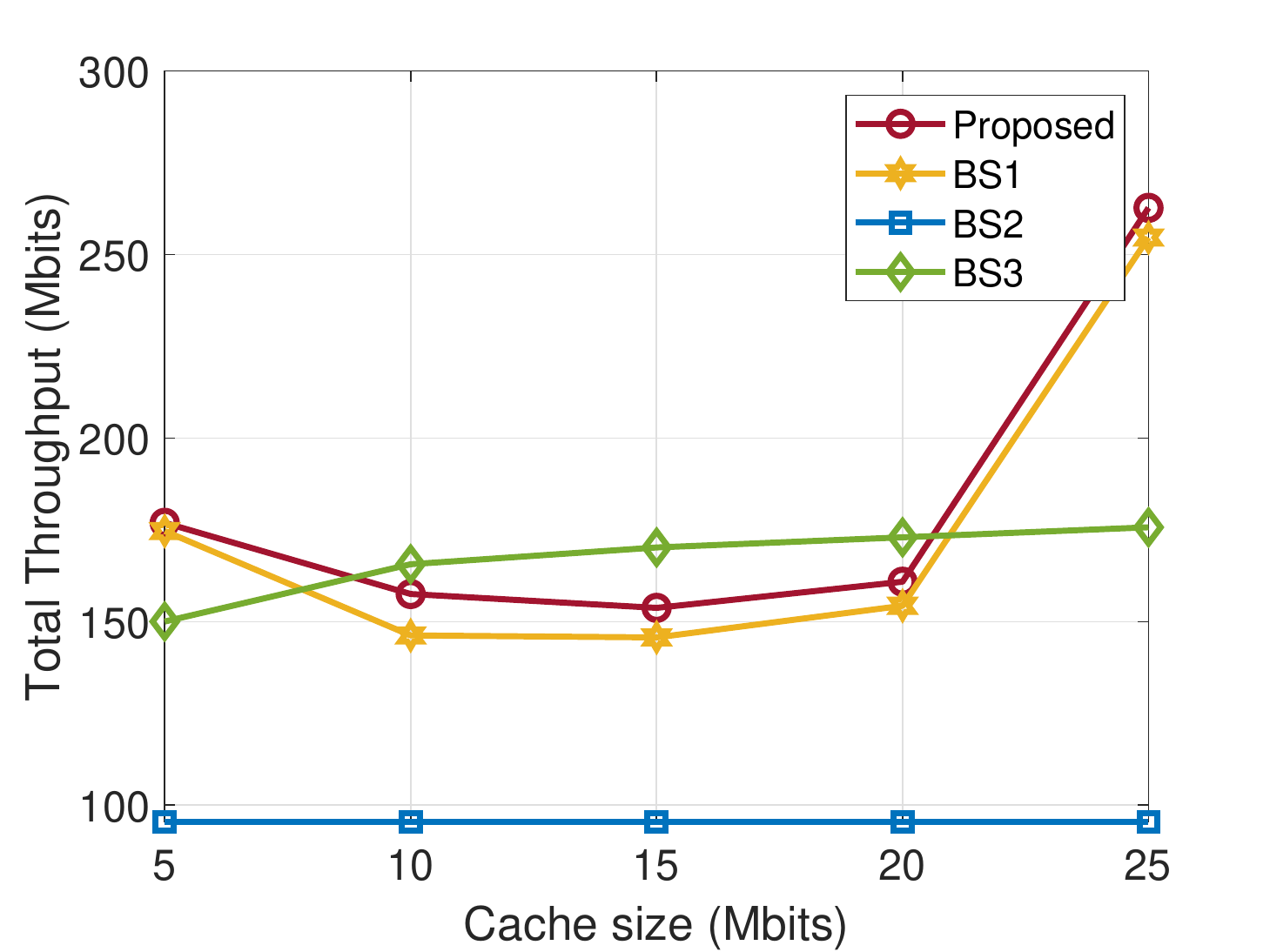}}
	\caption{{Max-min throughput and total throughput vs. cache size (Mbits).}}
	\label{fig:4}
\end{figure*}

\begin{figure*}[t]
	\centering   
	\subfigure[{Max-min throughput}] {\label{fig:5a}\includegraphics[width=9cm,height=7.5cm]{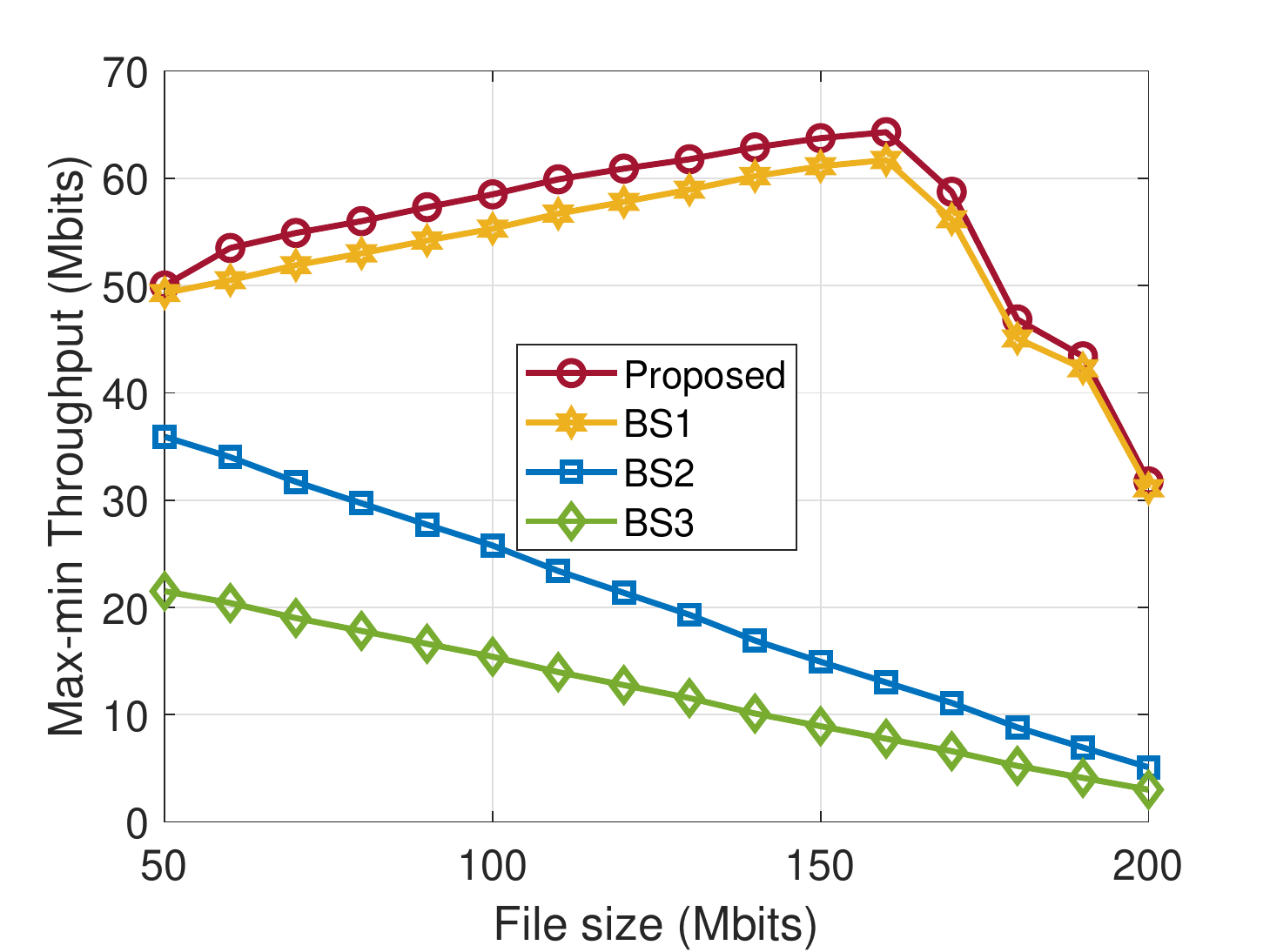}}
	\subfigure[{Total throughput}] {\label{fig:5b}\includegraphics[width=9cm,height=7.5cm]{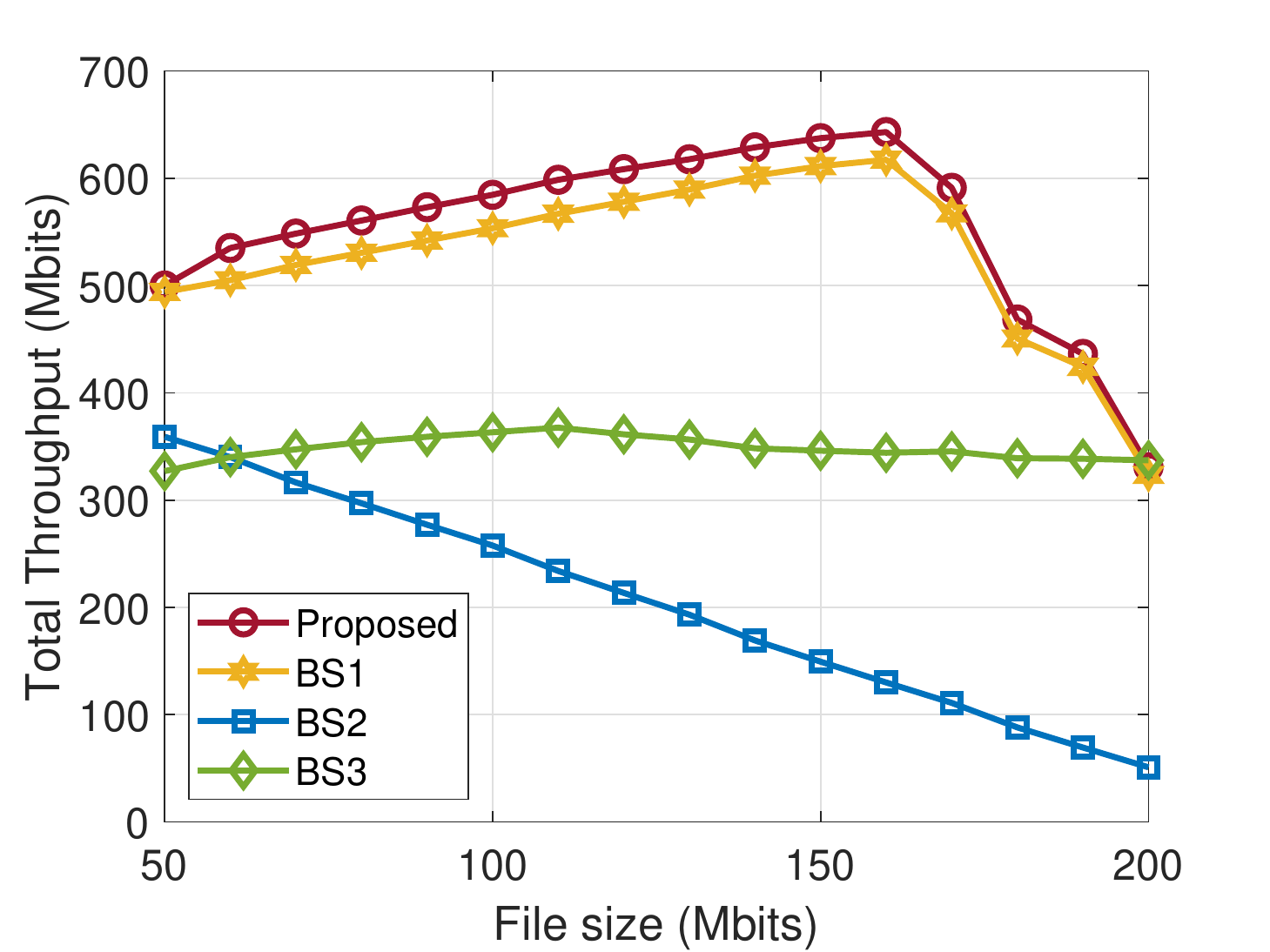}}
	\caption{Max-min throughput and total throughput vs. file size (Mbits).}
	\label{fig:5}
\end{figure*}

\begin{figure*}[t]
	\centering   
	\subfigure[{Max-min throughput}] {\label{fig:5a}\includegraphics[width=9cm,height=7.5cm]{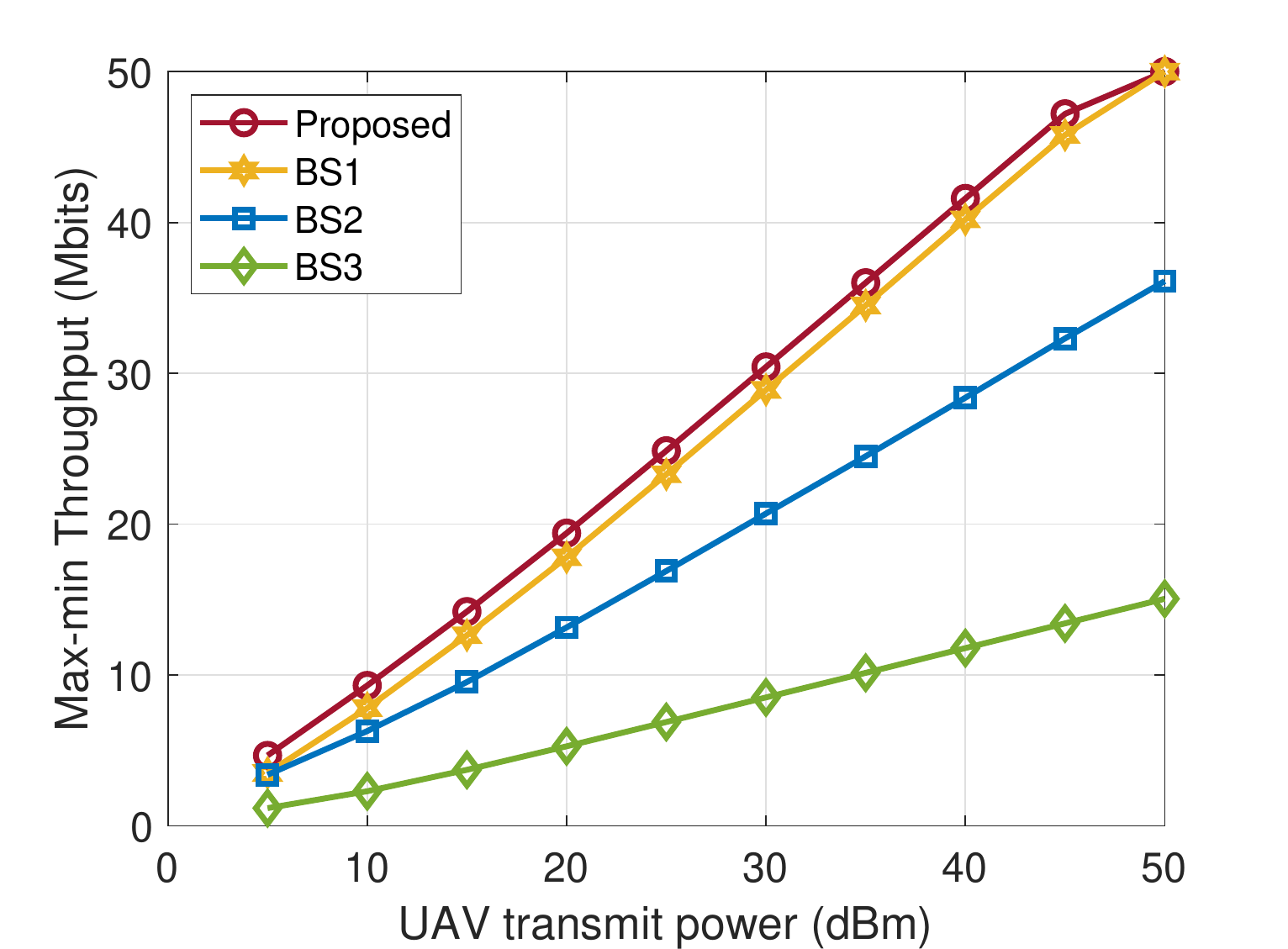}}
	\subfigure[{Total throughput}] {\label{fig:5b}\includegraphics[width=9cm,height=7.5cm]{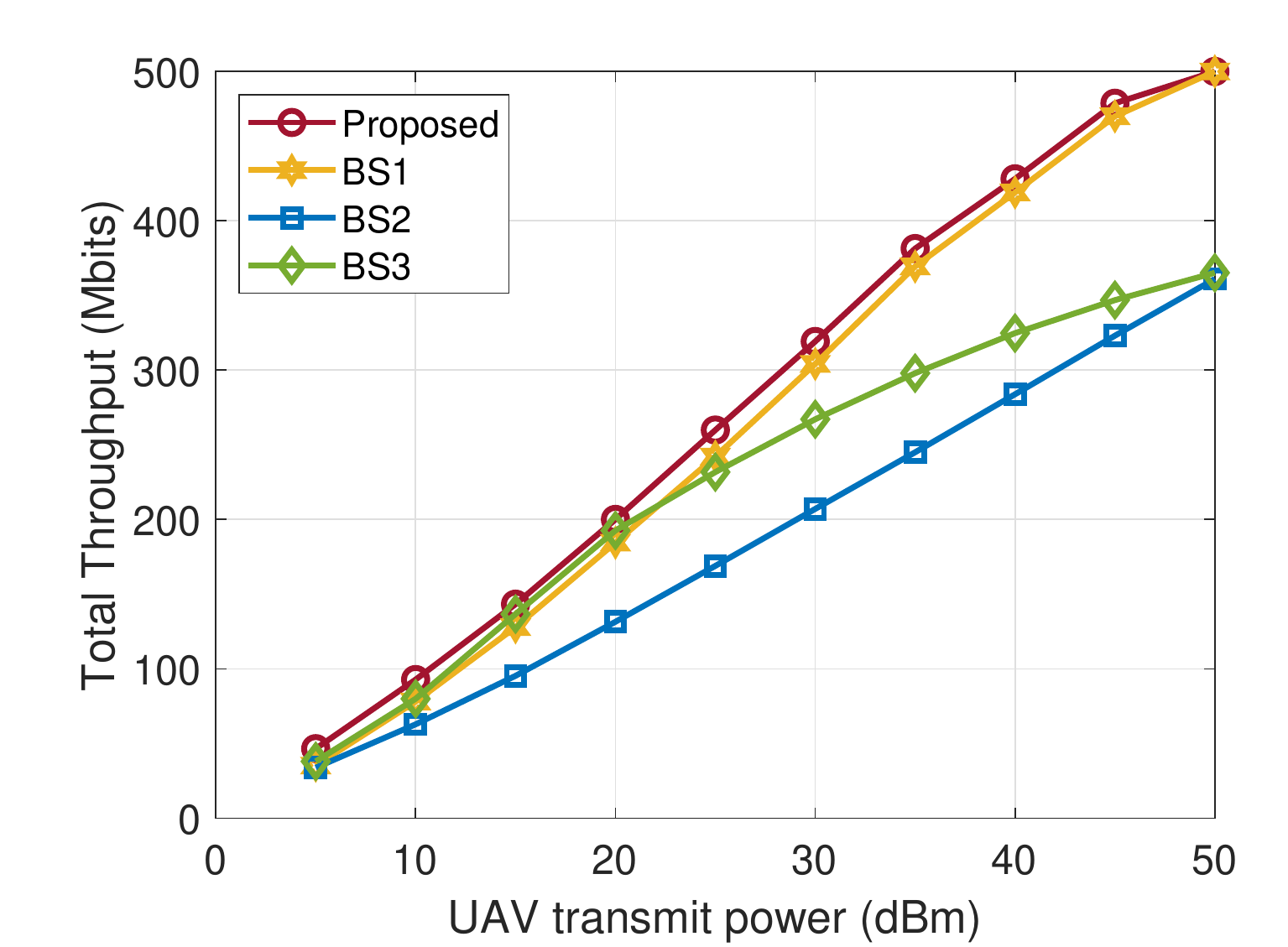}}
	\caption{{Max-min throughput and total throughput vs. UAV transmit power (dBm).} }
	\label{fig:6}
\end{figure*}

\section{Simulation Results}
\label{Sec:Num}
{In this section, numerical results are given to validate the proposed method, which jointly optimizes cache placement, resource allocation (i.e., bandwidth and transmit power), and the UAV trajectory design in satellite- and cache-aided UAV communication networks. We consider a system with one LEO satellite, one cache-enabled UAV, and $K$ GUs which is distributed in a horizontal plane, i.e., ${\rm Area}= x^2$ $(km^2)$, with $x=1$ km or 2 km.} We assume that the UAV's initial and final locations are respectively located at $\bq_{\rm I}=[1;0.7]$ km and $\bq_{\rm F}=[0.3;0]$ km. {The flight altitudes of the LEO satellite and the UAV are fixed at 2000 km and 1 km, respectively \cite{liou2008instability}.} The maximum bandwidth for the AL from  $u \to k$ is $B_{2k}=20$ Mhz. Therefore, the AWGN power is $\sigma^2=-174+10\log_{10}B\simeq -101$ dBm. {The maximum transmit powers of the satellite and the UAV are respectively set as $p_{1k}^n= 49.03$ dBm and $P_u^{\rm max}$ is ranging from 5 to 40 dBm \cite{LiSatUAV2020,JiUAVCache}.} Without other stated, other parameters are set as: path-loss exponent $\alpha=2$, the maximum UAV velocity $V_{\rm max}=50$ m/s, one time slot duration $\delta_t=0.5$ second, channel gain at the reference distance $\beta_0=$ -40 dB, total number of file $F= 30$ files, one file size $Q = 40$ Mbits, UAV's cache size $S=$ 10 files, Zipf skewness factor $\varrho$ = 0.8 \cite{ThangFDCache}. The error tolerance of iterative algorithm is set to $\epsilon=10^{-4}$. {The LEO satellite's orbital velocity is set to 6.9005 km/s based on \eqref{eq:sat_vel}.} The initial point of the LEO satellite is $[-345;0]$ km. {The penalty parameter $\kappa$
is initialized to 0.1 and incremented as $\kappa = 1.1 \kappa$ until $ \kappa \le 10$.} To show the superiority of our design, we compare the proposed scheme with the following benchmark schemes:
\begin{itemize}
	\item {Benchmark scheme 1 (BS1): UAV bandwidth and transmit power optimization with caching capability and fixed trajectory, i.e., a linear trajectory from initial to final locations \cite{JiUAVCache}.}
	
	\item Benchmark scheme 2 (BS2): UAV bandwidth, transmit power, and trajectory optimization without caching capability \cite{ThangFDCache}.
	
	\item Benchmark scheme 3 (BS3): UAV trajectory optimization with caching capability and fixed resource allocation, i.e., $b_{2k}^n \triangleq \frac{1}{K}$, $p_{2k}^n \triangleq \frac{P_u^{\rm max}}{K}$ \cite{JiUAVCache}.
\end{itemize}

{Fig. \ref{fig:2} plots the geometric distribution of GUs and the UAV trajectories for different traveling times, i.e., $N$ equals 40 and 100 time slots, with $B_{1k}=50$ Mbits. First, we observe that the UAV flies from the initial point to the furthest point where it can transmit information to GUs then back to the final point.} In contrast to the \cite{JiUAVCache} reference, in which the authors assume that the transmitter only serves up to one requester at a time slot, which is impractical and inefficient. In this work, we assume that the UAV can serve multiple GUs simultaneously to improve network performance, i.e., max-min throughput. {It leads to the fact that the UAV tends to fly to a point that keeps a relative distance to all GUs instead of flying to each GU's location as in \cite[Figure 7]{JiUAVCache}.} Further, it can be explained that if the UAV tries to fly closer to some GUs, thus it only helps to improve the throughput for these GUs while other users' performance is degraded. Thus, it does not guarantee the fairness between all GUs, which is the main purpose of this work. {Furthermore, we also find that by increasing the total flight time $T$, the UAV trajectory range becomes larger as it has more time to get closer to each GU, which improves the max-min throughput.}

{Fig. \ref{fig:3} describes max-min throughput and total throughput as functions of total flight time or the number of time slots $N$, where $B_{1k}$ = 50 Mbits, $P_u^{\rm max}$ = 15 dBm, $Q=$ 40 Mbits.} First, we observe that all schemes' performance increases significantly with larger values of the number of time slots. {That is because the higher the traveling time, the more data transmission rate per GU can be obtained.} Therefore, the minimum throughput value is improved. It can be seen that the proposed method always achieves the best performance as compared with other schemes. {Moreover, as the travel time is large enough, the performance of BS1 and BS2 methods can achieve the same performance as the proposed method. For example, the minimum throughput of BS1 and BS2 can reach 40 Mbits when the number of time slots is greater than 260 and 200, respectively; while BS3 always has the lowest value.} In Fig. \ref{fig:3b}, the total throughput is illustrated as a function of traveling time. We can see that Fig. \ref{fig:3b} has similar properties as Fig. \ref{fig:3a}. It shows the total throughput that the UAV successfully transfers to all GUs. {Similar to Fig. \ref{fig:3a}, when the number of time slot is lower than 80, the BS1 is better than the BS2 scheme.} Further, the proposed scheme is still the best one. {Specifically, the throughput performance of the proposed algorithm can serve up to 387 Mbits and the BS1 can achieve less than 26.46 $\%$, i.e., 284.6 Mbits, when $N=180$. In comparison, the BS2 and BS3 scheme imposes a 365 and 255.24 Mbits of total throughput, respectively.}

{In Fig. \ref{fig:4}, we study the influences of cache size, i.e., the number of files that can be stored at the UAV, on the network performance, where $Q=$ 60 Mbits, $N=$ 80, $B_{1k}=$ 20 Mbits, $x=$ 2 km. From the results, it can be seen that the proposed algorithm significantly enhances the minimum and total throughput compared to the benchmarks for all cache sizes. It is expected since the UAV has stored part of the requested files in their memory.} Thus, it does not need to demand from the satellite, which incurs more delay. As a result, the UAV has more time to communicate with GUs, and a higher data rate can be obtained. {For instance,  the minimum throughput of the proposed and BS1 schemes get 15.84 and 15.2 Mbits respectively at cache size equals 20 Mbits. Moreover, the BS2 and BS3 impose 9.55 and 5.61 Mbits, respectively. One more noticeable point in Fig. \ref{fig:4} is the performance of the BS2 scheme independent of the cache size values.} This is because the BS2 method is implemented without considering caching capability at the UAV. {Notably, when the cache size is ranging from 10 to 20, the BS3 method can achieve better total throughput compared to other schemes.}

{In Fig. \ref{fig:5}, we plot the max-min throughput as a function of the file size (in Mbits), where $N$ = 120, $x$ = 2 km. From the results, it is shown that the proposed scheme greatly improves the performance compared to the references for all file sizes. Specifically, at demanded data equals 160 Mbits, the max-min throughput value of proposed method is 64.4 Mbits, and the BS1 achieve less than 4.5 $\%$, i.e., 61.5 Mbits. Whereas the BS2 and BS3 impose 13 and 7.75, respectively. Notably, we also find that the performance of the proposed method and BS1 optains the maximum value at the optimal file size, then it will decreases. While the performance of BS2 and BS3 decreases dramatically. This shows the superiority of the resource and cache placement optimization in the proposed scheme and BS1 compared to BS2 and BS3. Nevertheless, for a given resource (i.e., bandwidth, transmit power, UAV speed, and total flight time), when the file size is too large (i.e., file size is larger than 160 Mbits), the performance of the proposed scheme and BS1 decreases significantly. This is due to the fact that the larger the file size, the more latency is required to transmit the requested data from the satellite to the UAV on the backhaul link. Therefore, the UAV has less time to transmit data to GUs. }

{Fig. \ref{fig:6} presents the results corresponding to the max-min throughput versus UAV transmit power $P_u^{\rm max}$, where $N=$ 100, $Q$ = 50 Mbits, $S=$ 10. As illustrated, system performance is enhanced by increasing the power budget, i.e., $P_u^{\rm max}$. That is due to the fact that the higher the transmit power, the higher the data transmission rate can be obtained, as shown in Eqs. \eqref{eq:8} and \eqref{eq:9}. Furthermore, the proposed scheme provides a better result in comparison with the benchmark ones when the transmit power is small, i.e., $P_u^{\rm max} \le 50$ dBm. Nevertheless, the BS1 method can obtain the same max-min throughput as the proposed method when $P_u^{\rm max}$ value is large, e.g., $P_u^{\rm max} \ge 50$ dBm. In this scenario, the UAV should operate in the BS1 scheme due to its simplicity and fast employment.} However, inherent restrictions of UAV is the limitation on size, weight, and power capability (SWAP). Thus, the proposed scheme is the best one that can adapt to all scenarios in practice.

\section{{Conclusion and Future Directions}}
\label{Sec:Con}
This paper studied LEO satellite- and cache-assisted UAV communications. Especially, we proposed a novel system model that jointly considers UAV, caching, and satellite communications in content delivery networks. In this context, we maximized the minimum achievable throughput among GUs via joint optimization of the cache placement, resource allocation, and UAV trajectory. Because the formulated problem was in the form of MINLP, it is difficult to solve directly. {Thus, we transformed the original problem into a solvable form using an alternative algorithm based on BCD method and SCA techniques. Extensive simulation results showed that our proposed algorithm improves up to 26.64$\%$, 79.79$\%$, and 87.96$\%$ in the max-min throughput compared to BS1, BS2, and BS3, respectively.} Notably, in the cases such as the UAV transmit power or the total traveling time size is large enough, the UAV should operate in fixed trajectory mode for a simple implementation. 

{The outcome of this work will motivate future works in satellite- and cache-assisted UAV wireless systems. One possible problem is to extend this work to a multi-UAV system, which imposes higher complexity but might further improve the network performance.}
{\section{Acknowledgement}
\label{ACK}
This research is supported by the Luxembourg National Research Fund under project FNR CORE ProCAST, grant C17/IS/11691338 and FNR CORE 5G-Sky, grant C19/IS/13713801.}


\bibliographystyle{IEEEtran}
\bibliography{IEEEfull}


\begin{IEEEbiography}[{\includegraphics[width=1in,height=1.25in,clip,keepaspectratio]{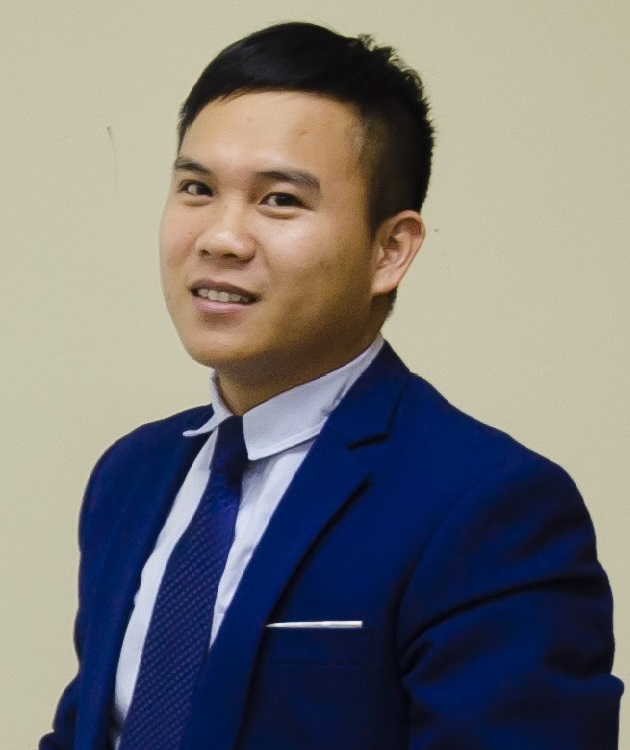}}]
	{Dinh-Hieu Tran,} (S'20) was born in 1989 in Vietnam, growing up up in Gia Lai. He received a BE degree in Electronics and Telecommunication Engineering Department from Ho Chi Minh City University of Technology, Vietnam, in 2012. In 2017, he completed an MS degree with honours in Electronics and Computer Engineering at Hongik University, South Korea. He is currently pursuing a PhD at the Interdisciplinary Centre for Security, Reliability and Trust (SnT), University of Luxembourg, under the supervision of Prof. Symeon Chatzinotas and Prof. Bj\"orn Ottersten. His research interests include UAVs, IoTs, mobile edge computing, caching, backscatter, B5G for wireless communication networks. He is a recipient of the IS3C 2016 best paper award.
\end{IEEEbiography}

\begin{IEEEbiography}[{\includegraphics[width=1in,height=2in,clip,keepaspectratio]{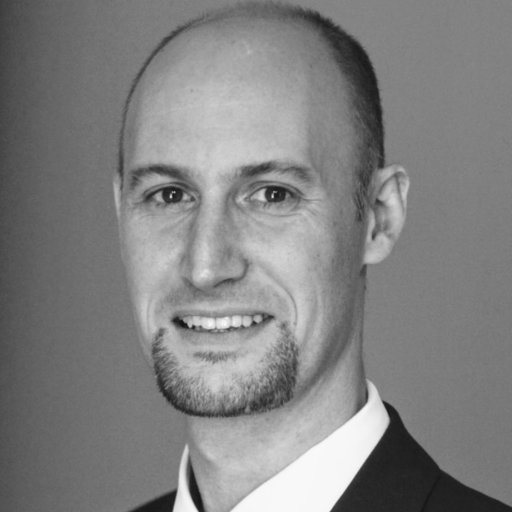}}]
	{Symeon Chatzinotas,} (S'06-M'09-SM'13) is currently Full Professor / Chief Scientist I in Satellite Communications and Head of the SIGCOM Research Group at SnT, University of Luxembourg. He coordinates research activities in communications and networking, acting as a PI in over 20 projects and is the main representative for 3GPP, ETSI, DVB.
	In the past, he worked as a Visiting Professor at the University of Parma, Italy, lecturing on 5G Wireless Networks. He was involved in numerous R$\&$D projects for NCSR Demokritos, CERTH Hellas and CCSR, University of Surrey.
	He was co-recipient of the 2014 IEEE Distinguished Contributions to Satellite Communications Award and Best Paper Awards at EURASIP JWCN, CROWNCOM, ICSSC. He has (co-)authored more than 450 technical papers in refereed international journals, conferences and scientific books.
	He is currently on the editorial board of the IEEE Transactions on Communications, IEEE Open Journal of Vehicular Technology, and the International Journal of Satellite Communications and Networking. 
\end{IEEEbiography}

\begin{IEEEbiography}[{\includegraphics[width=1in,height=1.25in,clip,keepaspectratio]{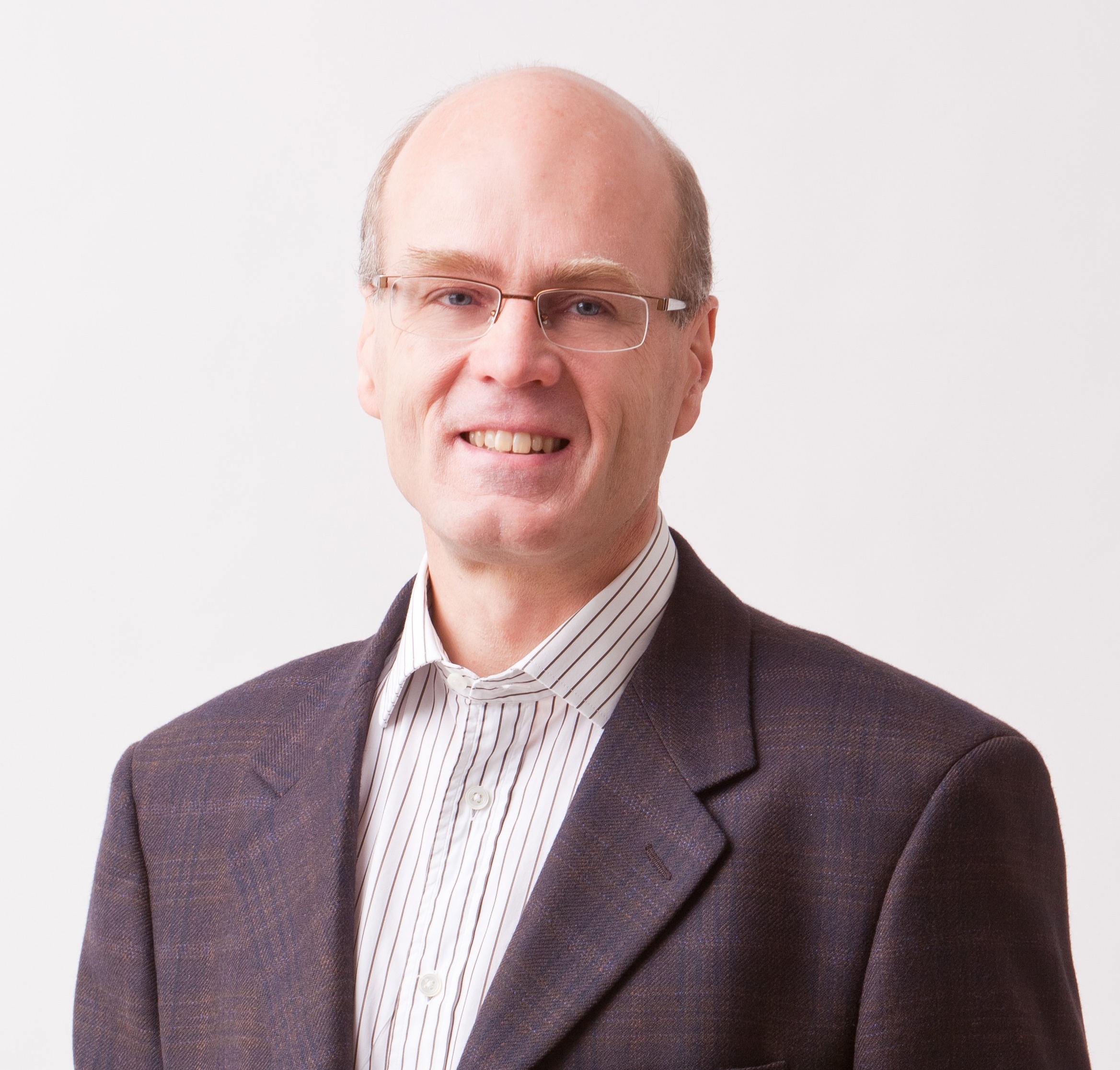}}]
	{Bj\"orn Ottersten,} (S'87-M'89-SM'99-F'04) was born in Stockholm, Sweden, in 1961. He received
	his MS degree in electrical engineering and applied physics from Linköping University, Linköping, Sweden, in 1986, and a PhD degree in electrical engineering from Stanford University, Stanford, CA, USA, in 1990. He has held research positions with the Department of Electrical Engineering at Linköping
	University, the Information Systems Laboratory at Stanford University, the Katholieke Universiteit Leuven in Belgium, and the University of Luxembourg. From 1996 to 1997, he was the Director of Research at ArrayComm, Inc., a start-up in San Jose, CA, USA, based on his patented technology. In 1991, he was appointed professor of signal processing at the Royal Institute of Technology (KTH) in Stockholm, Sweden. Dr. Ottersten has been Head of the Department for Signals, Sensors, and Systems, KTH, and Dean of the School of Electrical Engineering, KTH. He is currently the Director for the Interdisciplinary Centre for Security, Reliability and Trust at the University of Luxembourg.
	He is a recipient of the IEEE Signal Processing Society Technical Achievement Award and been twice awarded the European Research Council advanced research grant. He has co-authored journal papers which received the IEEE Signal Processing Society Best Paper Award in 1993, 2001, 2006, 2013, and 2019, and eight IEEE conference papers best paper awards. He has been a board member of IEEE Signal Processing Society and the Swedish Research Council and currently serves on the boards of EURASIP and the Swedish Foundation for Strategic Research. He has served as an Associate Editor for the IEEE TRANSACTIONS ON SIGNAL PROCESSING and the Editorial Board of the IEEE Signal Processing Magazine. He is currently a member of the editorial boards of the IEEE Open Journal of Signal Processing, EURASIP Signal Processing Journal, EURASIP Journal of Advanced Signal Processing and Foundations and Trends of Signal Processing. He is a fellow of EURASIP.	
\end{IEEEbiography}	
\end{document}